\documentclass[preprint,1p,authoryear]{elsarticle}

\usepackage{amssymb}

\usepackage[usenames,dvipsnames]{color}
\usepackage{graphicx}
\usepackage{array}
\usepackage{amsmath}
\usepackage{amssymb}
\usepackage{amsthm}
\usepackage[round]{natbib}
\usepackage[table,rgb]{xcolor}
\usepackage{cite}
\usepackage{multirow}
\usepackage{appendix}
\usepackage{pgfgantt}
\usepackage{pgfplots}
\usepackage{color}
\usepackage{tabu}
\usepackage{tikz}
\usepackage{tikz-timing}
\usepackage{tikz-qtree}
\usepackage{calc}
\usepackage[font=small,labelfont=bf]{caption}
\usepackage{bbding}
\usepackage[bookmarks=true, colorlinks=true, citecolor=blue, urlcolor=blue, linkcolor = blue]{hyperref}
\usepackage[colorinlistoftodos,textsize=tiny]{todonotes}
\usepackage{algorithm}
\usepackage{algorithmic}
\usepackage{afterpage}
\usepackage{float}
\usepackage{cleveref}
\usepackage{caption}
\usepackage{subcaption}
\usepackage{wrapfig}
\usepackage{lscape}
\usepackage{rotating}
\usepackage{hyphenat}
\usepackage{enumitem}
\usepackage[notquote]{hanging}
\usepackage[section]{placeins}
\usepackage{mathtools}
\usepackage[official]{eurosym}
\usepackage{fp}
\usepackage{environ}
\usetikzlibrary{patterns}

\newcommand{\hl}[2]{#1}
\newcommand{\hlim}[1]{#1}
\newcommand{\hli}[1]{#1}

\newtheorem{theorem}{Theorem}
\newtheorem{claim}{Claim}

\newtheorem{corollary}{Corollary}
\newtheorem{remark}{Remark}
\newtheorem{lemma}{Lemma}

\newtheorem{observation}{Observation}

\newcommand{\Cautoref}[1]{\hyperref[#1]{\Cref{#1}}}
\newcommand{\Cautorefs}[2]{\hyperref[#1]{#2 \ref{#1}}}
\DeclareMathOperator*{\argmin}{arg\,min}
\DeclareMathOperator*{\argmax}{arg\,max}

\usepackage{fancyhdr}
\allowdisplaybreaks

\tikzset{
	hatch distance/.store in=\hatchdistance,
	hatch distance=5pt,
	hatch thickness/.store in=\hatchthickness,
	hatch thickness=5pt
}
\pgfdeclarepatternformonly[\hatchdistance,\hatchthickness]{north east hatch}
{\pgfqpoint{-1pt}{-1pt}}
{\pgfqpoint{\hatchdistance}{\hatchdistance}}
{\pgfpoint{\hatchdistance-1pt}{\hatchdistance-1pt}}%
{
	\pgfsetcolor{\tikz@pattern@color}
	\pgfsetlinewidth{\hatchthickness}
	\pgfpathmoveto{\pgfqpoint{0pt}{0pt}}
	\pgfpathlineto{\pgfqpoint{\hatchdistance}{\hatchdistance}}
	\pgfusepath{stroke}
}

\begin{document}

\begin{frontmatter}


 \title{Single-machine scheduling with an external resource}

 \author[wupertal]{Dirk Briskorn}
 \author[SKEMA]{Morteza Davari\footnote[1]{Corresponding author}}
 \author[Leuven]{Jannik Matuschke}
 \address[wupertal]{Chair of Production and Logistics, University of Wuppertal, Germany}
\address[SKEMA]{SKEMA Business School, Universit\'{e} C\^{o}te d'Azur, France}
\address[Leuven]{Research Centre for Operations Management, KU Leuven, Belgium}

\begin{abstract}
This paper studies the complexity of single-machine scheduling with an external resource, which is rented for a non-interrupted period. Jobs that need this external resource are executed only when the external resource is available. There is a cost associated with the scheduling of jobs and a cost associated with the duration of the renting period of the external resource. \hli{We look at four classes of problems with an external resource: a class of problems where the renting period is budgeted and the scheduling cost needs to be minimized, a class of problems where the scheduling cost is budgeted and the renting period needs to be minimized, a class of two-objective problems where both, the renting period and the scheduling cost, are to be minimized, and a class of problems where a linear combination of the scheduling cost and the renting period is minimized.} We provide a thorough complexity analysis (NP-hardness proofs and  \hli{(pseudo-)polynomial} algorithms) for different members of these \hli{four} classes. 
\end{abstract}


\end{frontmatter}


\section{Introduction}

\hl{In the modern world where businesses face the challenge of the rapid growth of competitors, only those who lead their business in a more effective way manage to survive and reach success. Outsourcing, which has become a very popular trend, is one of the techniques that can help the business to reach advantages above opponents. This trend is very interesting because ``as more companies become involved in outsourcing, opportunities are opening up to owners of small and medium enterprises" \citep{Brown2005}. In the context of scheduling, outsourcing can be seen as using external resources. Examples of such resources are heavy or very expensive machinery, human experts, high-tech equipment, etc.}{R2}

We study single-machine scheduling problems with an external resource. We assume that the machine is available throughout the planning horizon and can process job at a time. However, some of the jobs require an external (and relatively expensive) resource (such as a crane, loader, human expert, etc.) and only one job can use the external resource at each moment in time. The external resource can be rented only once for an uninterrupted period. Both, scheduling of jobs and renting the external resource, incur costs. The renting cost is a linear function of the renting period. We investigate \hli{four} variants of problems. 
\begin{itemize}
	\item A variant where the renting cost is budgeted and the scheduling cost is to be minimized. 
	\item A variant where the scheduling cost is budgeted and the renting cost is to be minimized. 
	\item A two-objective variant where both, scheduling cost and renting cost, are to be minimized. 
	\item \hli{A variant where a linear combination of the scheduling cost is budgeted and the renting cost is to be minimized.} 
\end{itemize}
For the sake of simplicity, the two terms `external resource' and `resource' are used interchangeably in the remainder of this paper.

This paper is devoted to study the complexity of the above problems. Below, we review the complexity of a number of relevant classical scheduling problems. The single machine scheduling problem to minimize the total weighted completion time is polynomially solvable even with  a serial-parallel precedence graph \citep{Lawler1978} or with two-dimensional partial orders  \citep{AM09}.  On the other hand, single machine scheduling to minimize the total weighted completion time becomes strongly NP-hard as soon as release dates are present \citep{Lenstra1977}. Similarly, the single machine scheduling problem to minimize the maximum lateness is polynomially solvable even in the presence of precedence constraints \citep{Lawler1973} but strongly NP-hard with release dates \citep{Lenstra1977}. Moreover, the single machine scheduling problem with an objective function of weighted number of tardy jobs is known to be weakly NP-hard \citep{Karp:72,Lawler1969} whereas the problem with an objective function of total weighted tardiness is already strongly NP-hard \citep{Lenstra1977}. 

The literature on scheduling with external resources is rather scarce. The most relevant problem in the project scheduling literature is perhaps the \emph{resource renting problem} (RRP). The RRP as initially proposed by \citet{Nubel2001} aims to minimize the costs associated with renting resources throughout a project. These costs include fixed handling or procurement cost and variable renting cost. Unlike the setting in this paper where the external resource is rented for an uninterrupted period, the RRP considers renting resources that can be rented for as many as needed disjoint intervals. The problem has been recently extended by \citet{Vandenheede2016} who combined the RRP and the total adjustment cost problem and by \citep{Kerkhove2017} who studied a variant of the RRP with overtime. 

The RRP is closely associated with the \emph{resource availability cost problem} where resources are no longer to be rented but to be utilized. The assumption is that the resources, once utilized, are available for the whole duration of the project. The decision variables are the resource utilizations and the starting times. The utilization of a resource imposes some expenses in the resource availability cost problem, which needs to be minimized \citep{Rodrigues2010}. The problem is also known as the \emph{resource investment problem} \citep{Drexl2001}. To the best of our knowledge, the problem was first introduced by \citet{Mohring1984b} motivated by a bridge construction project. 

A budget imposed on the length of the renting period can be seen as a type of maximum delay constraints: For any pair of jobs $i, j$ that require the resource, the time between the start of $i$ and the end of $j$ must not exceed the budget. Such maximum delay constraints have been investigated by \citet{Wikum1994} in the context of the \emph{single-machine generalized precedence-constrained scheduling problem}. However, in this problem, maximum delay constraints are always combined with a non-negative minimum delay, thus enforcing an order among the two jobs. Thus, complexity results for this problem do not extend to external resource renting.

\hl{Two-objective scheduling problems, with one traditional objective and one non-traditional objective, are not new in the scheduling literature. For instance, 
\citet{Wan2010} study scheduling problems where the sum of a traditional cost-measure and a \emph{time-usage cost} is to be minimized. In this setting, a usage cost has to be paid for any time-slot in which the machine is active, with the usage cost varying over time.
The authors consider total completion time, maximum lateness/tardiness, total weighted number of tardy jobs, and total tardiness as traditional costs and identify special cases that can be solved efficiently. 
\citet{Chen2018} discuss a similar setting, but with the possibility to preempt jobs. They provide a polynomial time approximation scheme for minimizing the sum of time-usage costs and weighted total completion time.}{R1C2}

In this paper, we discuss the complexity of the \hli{four} classes of single-machine scheduling problems with an external resource and different objective functions. The remainder of this text is structured as follows: we formally define different variants of our problem in \Cautoref{sec:probdef}, discuss the complexity of \hli{these variants in \Cautorefs{sec:complex}{Sections}, \ref{sec:other} and \ref{sec:comb}} and finally summarize the results and discuss future research possibilities in \Cautoref{sec:conclusion}. 

\section{Problem definition}
\label{sec:probdef}

We consider a set $J$ of $n$ jobs with $p_j$, $d_j$, and $w_j$ representing the processing time, the due date, and the weight of job $j\in J$, respectively. Throughout the paper we assume $p_j$ and $w_j$ to be integer for each job $j\in J$. We let \mbox{$P =\sum_{j\in J} p_j$} and $W= \sum_{j\in J} w_j$ denote the sum of processing times and the sum of weights, respectively. There is a subset $J^r\subseteq J$ of jobs that require an external resource. We refer to jobs in $J^r$ as \emph{r-jobs} (resource jobs) and to jobs in $J^o =J\setminus J^r$ as \emph{o-jobs} (ordinary jobs). We assume that the external resource must be rented from the start of the first r-job to the completion of the last r-job. Let $C_j$ be the completion time of job $j$. The length of the renting period, which is denoted by $er$, is $er=\max_{j\in J^r}{C_j} - \min_{j\in J^r} \{C_j - p_j\}$. Now, for any single-machine scheduling problem $1||\gamma$ with objective function $\gamma$, there are \hli{four} natural counter-part problems with an external resource:
\begin{itemize}
	\item Problem $1|er|\gamma$ is to find a sequence of jobs that minimizes scheduling cost $\gamma$ among all sequences with a length of the renting period of at most $K^r$.
	
	\item Problem $1|\gamma|er$ is to find a sequence of jobs that minimizes the length of the renting period among all sequences with a scheduling cost of at most $K^{\gamma}$.
	
	\item Problem $1||(\gamma,er)$ is to find the sequences in the Pareto-front with respect to minimization of both, scheduling cost and the length of the renting period. 
	
	\item \hli{Problem $1||\gamma+er$ is to find the sequence that minimizes the sum of the scheduling cost and renting cost $\lambda\cdot er$ where $\lambda\geq 0$ denotes the renting cost per time unit.}
\end{itemize}

For a given sequence $\sigma$ of jobs we denote by $C^\sigma_j$ the completion time of job $j$, by $L^\sigma_j = C^\sigma_j - d_j$ its lateness and we let $U^\sigma_j$ be the indicator for tardiness, i.e., $U_j = 1$ if $C^\sigma_j>d_j$ and $U^\sigma_j = 0$ otherwise. We omit the superscript $\sigma$ whenever the sequence is clear from the context.

In what follows, we explore the complexity of the above problems when $\gamma$ is one of the following objective functions: 
\begin{itemize}
	\item Total completion time ($\sum C_j$)
	\item Total weighted completion time ($\sum w_jC_j$)
	\item Maximum lateness ($\max L_j $)
	\item Weighted number of tardy jobs ($\sum w_jU_j$)
\end{itemize}
\hli{Except $1||\sum w_jC_j + er$ and $1||\sum C_j + er$, which are polynomially solvable, all the other problems are shown to be NP-hard. Despite being NP-hard, they allow for pseudo-polynomial algorithms with the running time depending on the total processing time $P$ of all jobs. These results are described in \Cautorefs{sec:complex}{Sections}, \ref{sec:other} and \ref{sec:comb}.} A summary of \hl{classic results and}{R1C1}  complexity orders of our algorithms is given in \Cautoref{tbl:summery}.

\begin{table}[t]
	\centering
	\small
	\begin{tabular*}{\linewidth}{p{.2\linewidth}p{.2\linewidth}p{.6\linewidth}}
		\hline 
	 Problem	& $\gamma$ & Complexity  \\ 
		\hline 
		\multirow{4}{*}{ $1||\gamma$} & $\sum C_j$	& $\mathrm{O}(n \log n)$   \citep{Smith1956}\\		
		& $\sum w_jC_j$		& $\mathrm{O}(n \log n)$ \citep{Smith1956}\\	
		& $\max L_j$	& $\mathrm{O}(n \log n)$ \citep{Lawler1973} \\ 
		& $\sum w_jU_j$ & $\mathrm{O}(nP)$ \citep{Lawler1969} \\ \hline
		\multirow{4}{*}{ $1|er|\gamma$} & $\sum C_j$		& $\mathrm{O}(n^2P)$  \\		
		& $\sum w_jC_j$		& $\mathrm{O}(nP\min\{W,P\})$ \\	
		& $\max L_j$	& $\mathrm{O}(nP)$  \\ 
		& $\sum w_jU_j$ & $\mathrm{O}(nP^4)$\\ \hline 
		\multirow{4}{*}{ $1|\gamma|er$} & $\sum C_j$		& $\mathrm{O}(n^2P)$  \\		
		& $\sum w_jC_j$		& $\mathrm{O}(nP\min\{W,P\})$ \\	
		& $\max L_j$	& $\mathrm{O}(nP)$  \\ 
		& $\sum w_jU_j$ & $\mathrm{O}(nP^4\log P)$\\ \hline 
		\multirow{4}{*}{ $1||(\gamma,er)$} & $\sum C_j$		& $\mathrm{O}(nP^2)$  \\		
		& $\sum w_jC_j$		& $\mathrm{O}(nP^2)$ \\	
		& $\max L_j$	& $\mathrm{O}(nP^2)$  \\ 
		& $\sum w_jU_j$ & $\mathrm{O}(nP^5)$\\ \hline 
		\multirow{4}{*}{ $1||\gamma+er$} & $\sum C_j$		& $\mathrm{O}(n\log n)$  \\		
		& $\sum w_jC_j$		& $\mathrm{O}(n\log n)$ \\	
		& $\max L_j$	& $\mathrm{O}(nP^2)$  \\ 
		& $\sum w_jU_j$ & $\mathrm{O}(nP^5)$\\ \hline 
	\end{tabular*} 
	\caption{Summary of classic results and complexity orders of our algorithms}
	\label{tbl:summery}
\end{table}

Note that these results suggest that all variants are polynomially solvable whenever the total processing time $P$ is polynomial in the number $n$ of jobs. Note, furthermore, that the case with identical processing times, that is $p_j=p$ for each job $j$, can be easily reduced to the case with $p=1$. The latter is solvable in polynomial time as $P = n$ in this case. Hence, each problem is solvable in polynomial time under identical processing times.

 
\section{Complexity results for $1|er|\gamma$}
\label{sec:complex}

In this section, we discuss the complexity of $1|er|\gamma$. Throughout this section, we sometimes refer to a sequence $\sigma$ as \emph{feasible} which means  it respects the resource budget.

Without loss of generality, we assume $J = \{1, \dots, n\}$ (later, we will assume this numbering to reflect an ordering of the jobs according to some attribute) and define $J[a, b] := \{j \in J : a \leq j \leq b\}$ for $a, b \in J$ .

We denote the total processing time of a job set $S$ by $p(S) = \sum_{j\in S} p_j$ and its total weight by $w(S) = \sum_{j\in S} w_j$. Also, we use $\mathrm{TWC}(\sigma)$ as the total weighted completion time and $L_{\max}(\sigma)$ as the maximum lateness for sequence $\sigma$. Note that, to avoid excess of notations, we let $\sigma$ not only represent a sequence, but also imply the sequence's set of jobs. Thus, $p(\sigma)$ denotes to total processing time of jobs present in $\sigma$.

\subsection{Total weighted completion time}
\label{subsec:twc}

In this section, we review the complexity of $1|er|\sum w_j C_j$. We first prove that even the unweighted problem $1|er|\sum C_j$ is already NP-hard (\Cautoref{totalcompletion}) and then we propose a pseudo-polynomial algorithm for $1|er|\sum w_j C_j$ (\Cautoref{totalwc}). 

\begin{theorem}
	\label{totalcompletion}
	$1|er|\sum C_j$ is NP-hard.
\end{theorem}

\begin{proof}
	We prove the NP-hardness of $1|er|\sum C_j$ by a reduction from \textsc{Even-Odd-Partition} which is known to be NP-hard, see \citep{Garey1988}.
	
\textsc{Even-Odd-Partition}: Given integers $a_1,\ldots,a_{2m}$ with $a_{k-1}<a_{k}$ for $k = 2, \dots ,2m$ and with total value $2B$, is there a subset of $m$ of these numbers with total value of $B$ such that for each $k=1,...,m$ exactly one of the  pair $\{a_{2k-1},a_{2k}\}$ is in the subset? 

In the following, we assume $B \ge 2m(m+1) -2$ for our instance. Note that we can always avoid $B < 2m(m+1) -2$ by increasing the value of each integer $a_k$, $k=1,\ldots,2m$, by $2(m+1)$ and increasing the value of $B$ by $2m(m+1)$, accordingly.
	
	Given such an instance of \textsc{Even-Odd-Partition}, we construct an instance of $1|er|\sum C_j$ with $2m+2$ jobs as follows:
	\begin{itemize}
		\item $J=\{1,\ldots,2m+2\}$ and $J^r=\{2m+1,2m+2\}$, 
		\item $p_j=B^2+a_j$ for each $j=1,\ldots,2m$, 
		\item $p_{2m+1}=0$ and $p_{2m+2}=C+D+1$, and
		\item $K^r=p_{2m+2}+mB^2+B$
	\end{itemize}
	where 
	$$C=\sum_{k=1}^m(m+1-k)(p_{2k-1}+p_{2k})+(mB^2+B)(m+1)$$
 and 
	 $$D=\sum_{j=1}^{2m}p_j = 2mB^2 + 2B.$$
	 
	We claim that there is a feasible schedule with total completion time of no more than $$2(C+D)+1$$ if and only if the answer to the instance of 
	\textsc{Even-Odd-Partition} 	
	is yes.
	
	First, consider a job sequence $\sigma$ with total completion time of no more than $2(C+D)+1$.
	
	\begin{claim}
	Job $2m+2$ is the last job in $\sigma$ and job $2m+1$ is not started before $mB^2+B$.
	\end{claim}
	\begin{proof}
	    Assume that job $2m+2$ is not the last job. Then, at least two jobs have a completion time of at least $p_{2m+2}=C+D+1$ and, thus, total completion time is at least $2(C+D)+2$. 
	    
	    Due to $C_{2m+2}^{\sigma}=\sum_{j=1}^{2m+2}p_j$ and due to feasibility of $\sigma$, job $2m+1$ is not started before
	    \begin{align*}
	        D+p_{2m+1}+p_{2m+2}-(p_{2m+2}+mB^2+B)=D-mB^2-B=mB^2+B.
	    \end{align*}
	\end{proof}

	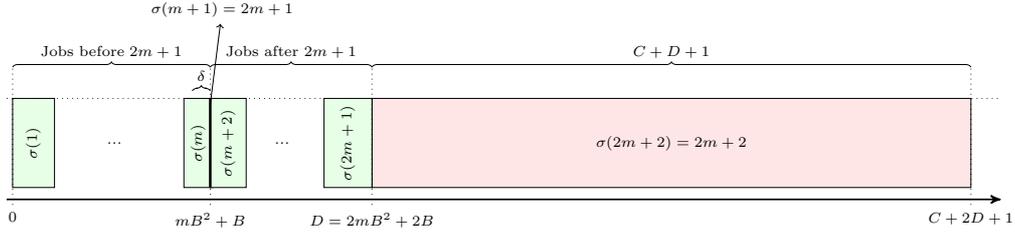
\begin{figure*}
		
		\centering
		\resizebox{\linewidth}{!}{%
		
		\begin{tikzpicture}[
		axis/.style={thick, ->, >=stealth'},
		important line/.style={very thick},
		dashed line/.style={dashed,  thick},
		every node/.style={color=black,}
		]

		\draw[dotted] (-.1,1.5)  -- (16.5,1.5);	
		
		\node[fill=red!10, rectangle, draw=black, minimum height=1.5cm,minimum width=10cm] at (11,0.75) {\scriptsize$\sigma(2m+2) =2m+2$};
		\node[fill=green!10, rectangle, draw=black, minimum height=1.5cm,minimum width=.7cm, label=] at (.35,0.75) {\rotatebox[]{90}{\scriptsize$\sigma(1)$}};
		\node[fill=green!10, rectangle, draw=black, minimum height=1.5cm,minimum width=.4cm, label=] at (3.1,0.75) {\rotatebox[]{90}{\scriptsize$\sigma(m)$}};
		\node[fill=green!10, rectangle, draw=black, minimum height=1.5cm,minimum width=.6cm, label=] at (3.6,0.75) {\rotatebox[]{90}{\scriptsize$\sigma(m+2)$}};
		\node[fill=green!10, rectangle, draw=black, minimum height=1.5cm,minimum width=.8cm, label=] at (5.6,0.75) {\rotatebox[]{90}{\scriptsize$\sigma(2m+1)$}};
		
		\node[] at (1.7,0.75) {{\scriptsize$...$}};
		\node[] at (4.5,0.75) {{\scriptsize$...$}};
		
		\node[] (D1) at (3.5,3) {{\scriptsize$\sigma(m+1) =2m+1$}};
		\draw[<-] (D1) -- (3.3,1.4);
		
		\draw[important line] (3.3,1.5) -- (3.3,0.0);
		
		\draw[dotted] (0,2) -- (0,-0.3) node(yline)[below] {\scriptsize$0$};
		\draw[dotted] (3.3,2) -- (3.3,-0.3) node(yline)[below] {\scriptsize$mB^2+B$};
		\draw[dotted] (6,2) -- (6,-0.3) node(yline)[below] {\scriptsize$D=2mB^2+2B$};
		\draw[dotted] (16,2) -- (16,-0.3) node(yline)[below] {\scriptsize$C+2D+1$};

		\draw[decoration={brace,raise=5pt},decorate]
		(3,1.45) -- node[above=6pt] {\scriptsize$\delta$} (3.3,1.45);
		\draw[decoration={brace,raise=5pt},decorate]
		(0,1.85) -- node[above=6pt] {\scriptsize Jobs before $2m+1$} (3.3,1.85);
		\draw[decoration={brace,raise=5pt},decorate]
		(3.3,1.85) -- node[above=6pt] {\scriptsize Jobs after $2m+1$} (6,1.85);
		\draw[decoration={brace,raise=5pt},decorate]
		(6,1.85) -- node[above=6pt] {\scriptsize $C+D+1$} (16,1.85);
		
		\draw[axis] (-.1,-.2)  -- (16.5,-.2);
		
		\end{tikzpicture}
}

		\caption{The schedule for sequence $\sigma$ in the proof of \Cautoref{totalcompletion}}
		\label{fig:seqsigma}
	\end{figure*}
	
	\begin{claim}
	Exactly $m$ jobs are scheduled between $2m+1$ and $2m+2$ in $\sigma$.
	\end{claim}
	\begin{proof}
	On the one hand, no more than $m$ jobs can be scheduled between $2m+1$ and $2m+2$ since total processing time of the jobs following job $2m+1$ for any $B>1$ amounts to at least
	$$p_{2m+2} + (m+1)B^2 > p_{2m+2} + mB^2+B = K^r.$$
	On the other hand, if less than $m$ jobs are scheduled between $m+1$ and $m+2$, the total processing time $\mathrm{TC}(\sigma)$ exceeds $2(C+P)+1$ since
	\begin{flalign*}
	\mathrm{TC}(\sigma)>&\underbrace{\sum_{j=1}^{2m}jB^2+2B}_{\text{A LB for }\sum_{j=1}^{2m}C_j}+\underbrace{(m+1)B^2}_{\text{A LB for }C_{2m+1}}+\underbrace{C_{\max}}_{C_{2m+2}}\\
	=\quad&\sum_{j=1}^{m}(m+1-j)B^2+\sum_{j=1}^{m}(m+m+1-j)B^2+(m+1)B^2\\ 
	& +2B+C_{\max}\\
	=\quad&\sum_{j=1}^{m}2(m+1-j)B^2+(mB^2)(m+1)+B^2+2B+C_{\max}\\
	=\quad&\sum_{j=1}^{m}2(m+1-j)(B^2+2B)+(mB^2)(m+1)\\
	&+\underbrace{B^2 +2B - 2m(m+1)B }_{\ge 0 \quad (\text{since } B \ge 2m(m+1) -2)}\ +\ C_{\max}\\
	\geq\quad&\sum_{j=1}^{m}2(m+1-j)(B^2+2B)+(mB^2)(m+1)+C_{\max}\\
	=\quad&\underbrace{\sum_{j=1}^{m}2(m+1-j)(B^2+B)+(mB^2+mB)(m+1)}_{> C}\ +\ C_{\max}\\
	>\quad&C+C_{\max}=C+D+p_{2m+2}=C+D+C+D+1 \\
	=\quad& 2(C+D)+1. 
	\end{flalign*}%
    \end{proof}
    
	Following the above two claims, we conclude that $\sigma{(m+1)} = 2m+1$ and $\sigma{(2m+2)} = 2m+2$. The schedule for sequence $\sigma$ is depicted in \Cautoref{fig:seqsigma}. We derive the total completion time $\textrm{TC}(\sigma)$ of $\sigma$ as follows:
	\begin{flalign*} 
	\textrm{TC}(\sigma) &\ =   \sum_{k = 1}^{m} {C_{\sigma(k)}} \ +\underbrace{mB^2+B+\delta}_{C_{2m+1}}+ \sum_{k = m+1}^{2m+1} {C_{\sigma(k)}} +\underbrace{C_{\max}}_{C_{2m+2}} 
	\end{flalign*}
	where $$\delta=\sum_{k=1}^m p_{\sigma(k)}-(mB^2+B)$$ is the difference between the starting time of job $2m+1$ according to $\sigma$ and its earliest starting time. Note that $\delta\geq 0$ due to feasibility of $\sigma$. Since $C_{\sigma(k)}=\sum_{s=1}^{k} p_{\sigma(s)}$, $\textrm{TC}(\sigma)$ can be rewritten as 
	\begin{flalign*}
	\textrm{TC}(\sigma) = &\  \sum_{k=1}^{m} {(2m+1-k)p_{\sigma(k)}}+\sum_{k=1}^m(m+1-k)p_{\sigma(k+m+1)}\\
	&+{mB^2+B+\delta}+{C_{\max}} \\
	= &\  \sum_{k=1}^m(m+1-k)p_{\sigma(k)}+\sum_{k=1}^m(m+1-k)p_{\sigma(k+m+1)} \\ &+(m+1)(mB^2+B+\delta)+C_{\max}\\
	=&\ \sum_{k=1}^m(m+1-k)\left(p_{\sigma(k)}+p_{\sigma(k+m+1)}\right) +(m+1)\delta\\
	&+\underbrace{(m+1)(mB^2+B)+C+2D+1}_{\text{constant}}.
	\end{flalign*}
	
We observe that $\textrm{TC}(\sigma) \le 2(C+D)+1$ only if $$\sum_{k=1}^m(m+1-k)\left(p_{\sigma(k)}+p_{\sigma(k+m+1)}\right) + (m+1)\delta \le \sum_{k=1}^m(m+1-k)(p_{2k-1}+p_{2k})$$ holds. This inequality holds only if $\delta = 0$ and for each $k = 1,\ldots,m$, one of the jobs $2k$ or $2k-1$ is assigned to position $k$ and the other to position $k+m+1$ in $\sigma$ (recall that numbers are ordered increasingly in \textsc{Even-Odd-Partition}). Thus, we conclude that the subsets of jobs before and after job $2m+1$ constitute a yes-certificate for the corresponding instance of \textsc{Even-Odd-Partition}.

Second, if a yes-certificate for the instance of \textsc{Even-Odd-Partition} is given we can construct a sequence with the structure discussed above and, thus, yielding total completion time of at most $2(C+D)+1$. This completes the proof.
\end{proof}

We now show that $1|er|\sum w_jC_j$ can be solved in pseudo-polynomial time. We show, in \Cautoref{lem:twc-opt-structure}, that there always exists an optimal sequence with a special structure consisting of five blocks that are internally ordered according to the \emph{weighted shortest processing time} (WSPT) rule and then we exploit this structure to find an optimal sequence using \emph{dynamic program}s (DPs) in \Cautoref{lem:twc_nP2} and \Cautoref{lem:twc_nPW}.

\hl{
Without loss of generality, we assume the jobs to be numbered according to WSPT (i.e., $J = \{1, \dots, n\}$ with $w_1 / p_1 \geq \dots \geq w_n / p_n$).  We let $\alpha = \min J^r$ and $\beta = \max J^r$ be the r-jobs with lowest and highest WSPT index, respectively, and define $H = J[\alpha, \beta] \cap J^o$ as the set of o-jobs whose WSPT index is between $\alpha$ and $\beta$. For $j \in J$, we further let $t_j = p(J[1,j-1])$ be the total processing time of the jobs preceding $j$ in WSPT order.

For any two job sets $X, Y \subseteq H$ with $X \cap Y = \emptyset$, we define a corresponding sequence $\sigma_{X, Y}$ as follows. The sequence consists of five blocks and within each block the jobs are sorted by increasing WSPT index. The first block is $J[1, \alpha-1]$; the second block is $X$; the third block is $J[\alpha, \beta] \setminus (X \cup Y)$; the fourth block is $Y$; the fifth block is $J[\beta+1, n]$. \Cautoref{fig:fiveblock} depicts such a sequence.%
}{R2C2}

\begin{figure}
	\centering
	\begin{tikzpicture}[
		axis/.style={thick, ->, >=stealth'},
		important line/.style={very thick},
		dashed line/.style={dashed,  thick},
		every node/.style={color=black,}
		]

		\node[fill=green!10, rectangle, draw=black, minimum height=.5cm,minimum width=2cm] at (1,0.5) {\scriptsize$J[1,\alpha-1]$};
		\node[fill=blue!10, rectangle, draw=black, minimum height=.5cm,minimum width=.5cm] (out1) at (2.25,0.5) {\scriptsize$X$};
		\node[fill=blue!10, rectangle, draw=black, minimum height=.5cm,minimum width=.7cm] (out2) at (5.15,0.5) {\scriptsize$Y$};
		\node[preaction={fill=green!10}, fill=green!10,  pattern=north west lines, pattern color=red!10, rectangle, draw=black, minimum height=.5cm,minimum width=2.3cm] at (3.65,0.5) {\scriptsize$J[\alpha,\beta]\setminus (X\cup Y)$};
		\node[fill=green!10, rectangle, draw=black, minimum height=.5cm,minimum width=2.5cm] at (6.75,0.5) {\scriptsize$J[\beta+1,n]$};
		
		\draw[dotted] (0,1) -- (0,-0.1) node(yline)[below] {\scriptsize$0$};
		\draw[dotted] (8,1) -- (8,-0.1) node(yline)[below] {\scriptsize$P$};

		\draw[axis] (-.1,0)  -- (8.5,0);
		\end{tikzpicture}
		\caption{sequence $\sigma_{X,Y}$}
	\label{fig:fiveblock}
\end{figure}
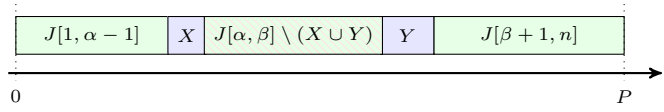

\begin{lemma}\label{lem:twc-opt-structure}
	For each instance of $1|er|\sum w_jC_j$ there exists $X^*, Y^* \subseteq H$ with $\max X^* < \min Y^*$ such that $\sigma_{X^*, Y^*}$ is optimal.
\end{lemma}

\begin{proof}
	Let $\sigma$ be an optimal feasible sequence. Let $i_1 := \min \{i : \sigma(i) \in J^r\}$ and $i_2 := \max \{i : \sigma(i) \in J^r\}$ be the first and last occurrence, respectively, of an r-job in the sequence. Let $S = \{\sigma(i_1), \dots, \sigma(i_2)\}$. Note that $\sum_{j \in S} p_j \leq K^r$ by feasibility of $\sigma$. Hence, any sequence that schedules the jobs of $S$ consecutively is feasible. 	In particular, rearranging the jobs within $S$ according to WSPT maintains feasibility of $\sigma$ without increasing total weighted completion time. Therefore, we can assume, without loss of generality,~$\alpha = \sigma(i_1) < \dots < \sigma(i_2) = \beta$ (i.e., the jobs in $S$ are scheduled according to WSPT and, in particular, $S \subseteq J[\alpha, \beta]$).
	
Now consider the job set $J' := J \setminus S \cup \{j'\}$ where the jobs of $S$ are merged into the single job $j'$ with processing time $p_{j'} = p(S)$ and weight $w_{j'} = w(S)$. Let $$X^* := \{j \in H \setminus S : w_j / p_j \geq w_{j'} / p_{j'}\}$$ and $$Y^* := H \setminus (S \cup X^*).$$ Note that $\max X^* < \min Y^*$ by construction and that $\sigma_{X^*, Y^*}$ is a feasible sequence for $J$. Further note that both $\sigma_{X^*, Y^*}$ and $\sigma$ induce sequences $\sigma'_{X^*, Y^*}$ and $\sigma'$ for $J'$, respectively. In particular, $\sigma'_{X^*, Y^*}$ orders jobs in $J'$ according to WSPT and therefore $$\mathrm{TWC}(\sigma'_{X^*, Y^*}) \leq \mathrm{TWC}(\sigma').$$ Moreover, 
	\begin{align*}
	\mathrm{TWC}(\sigma_{X^*, Y^*}) & = \mathrm{TWC}(\sigma'_{X^*, Y^*}) - \sum_{j \in S} w_j \cdot p(S[j+1, n] )\\
	& \leq \mathrm{TWC}(\sigma') - \sum_{j \in S} w_j \cdot p(S[j+1, n] ) \\
	&= \mathrm{TWC}(\sigma),
	\end{align*}
which establishes that $\sigma_{X^*, Y^*}$ is also an optimal sequence for $J$.
\end{proof}

\begin{lemma}
	\label{lem:twc_nP2}
	$1|er|\sum w_jC_j$ can be solved in $\mathrm{O}(nP^2)$-time.
\end{lemma}

\begin{proof}
For each $\kappa, \rho \in \mathbb{N}$ with $\alpha < \kappa \leq \beta$ and $\rho \le K^r$, let us define 
\begin{align*}
\cal{X}_{\kappa,\rho} &= \{X \subseteq H : \max X < \kappa, p(X) = \rho \}, \\ 
\cal{Y}_{\kappa,\rho} &= \{Y \subseteq H : \min Y \ge \kappa, p(Y) = \rho \}, \\
f_\kappa(X) &= \sum_{j = \alpha}^{\kappa - 1} w_j C^{\sigma_{X,\emptyset}}_{j} \quad \text{ and } \quad g_\kappa(Y) = \sum_{j= \kappa }^{\beta} w_{j}C^{\sigma_{\emptyset,Y}}_{j}. 
\end{align*}
Also let 
\begin{alignat*}{3}
X_{\kappa,\rho} &\in \argmin_{X \in \mathcal{X}_{\kappa,\rho}} \{f_\kappa(X)\} & \quad \text{ and } \quad & \bar{X}_{\kappa,\rho} = J[\alpha,\kappa-1] \setminus X_{\kappa,\rho},\\ Y_{\kappa,\rho} &\in \argmin_{Y \in \mathcal{Y}_{\kappa,\rho}} \{g_\kappa(Y)\} & \quad \text{ and } \quad & \bar{Y}_{\kappa,\rho} = J[\kappa,\beta] \setminus Y_{\kappa,\rho}.
\end{alignat*}

Based on \Cautoref{lem:twc-opt-structure}, it suffices to find $X^*, Y^* \subseteq H$ with $\max X^* < \min Y^*$ such that $\sigma_{X^*, Y^*}$ is optimal. The first step is thus to compute ${X}_{\kappa,\rho}$ and ${Y}_{\kappa,\rho}$ for all pairs $(\kappa,\rho)$ and then compute $X^* = X_{\kappa^*,\rho_1^*}$ and $Y^* = Y_{\kappa^*,\rho_2^*}$, where 
\begin{align*}
(\kappa^*,\rho_1^*,\rho_2^*) \in \argmin_{(\kappa,\rho_1,\rho_2) \in \Xi} \{f_\kappa(X_{\kappa,\rho_1})+ g_\kappa(Y_{\kappa,\rho_2})\}
\end{align*}
and 
\begin{align*}
\Xi = \{(\kappa,\rho_1,\rho_2) \;|\; {\cal{X}}_{\kappa,\rho_{1}},{\cal{Y}}_{\kappa,\rho_2} \neq \emptyset, \; p(J[\alpha,\beta]) - \rho_1 - \rho_2 \le K^r \}.
\end{align*}
Given a tuple $(\kappa,\rho_1,\rho_2)$, \Cautoref{fig:sequencefortuple} depicts the associated sequence $\sigma_{X_{\kappa,{\rho_1}},Y_{\kappa,{\rho_2}}}$.

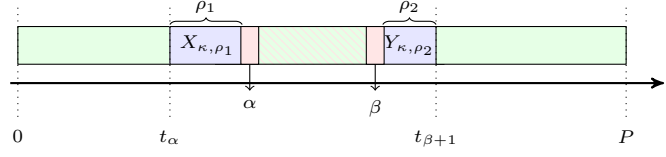
\begin{figure}
	\centering
	\begin{tikzpicture}[
	axis/.style={thick, ->, >=stealth'},
	important line/.style={very thick},
	dashed line/.style={dashed,  thick},
	every node/.style={color=black,}
	]

	\node[fill=green!10, rectangle, draw=black, minimum height=.5cm,minimum width=2cm] at (1,0.5) {};
	\node[preaction={fill=green!10}, fill=green!10,  pattern=north west lines, pattern color=red!10, rectangle, draw=black, minimum height=.5cm,minimum width=2.1cm] at (3.65,0.5) {};
	\node[fill=blue!10, rectangle, draw=black, minimum height=.5cm,minimum width=1cm] (out1) at (2.5,0.5) {\scriptsize$X_{\kappa,{\rho_1}}$};
	\node[fill=blue!10, rectangle, draw=black, minimum height=.5cm,minimum width=.7cm] (out2) at (5.15,0.5) {\scriptsize$Y_{\kappa,{\rho_2}}$};

	\node[fill=green!10, rectangle, draw=black, minimum height=.5cm,minimum width=2.5cm] at (6.75,0.5) {};
	\node[fill=red!10, draw=black, rectangle, minimum height=.5cm,minimum width=.2cm] (n1) at (3.05,0.5) {};
	\node[fill=red!10,rectangle, draw=black,  minimum height=.5cm,minimum width=.2cm] (nk2) at (4.7,0.5) {};
	
	\draw[->] (n1) -- (3.05,-0.1) node(yline)[below] {\scriptsize$\alpha$};
	\draw[->] (nk2) -- (4.7,-0.1) node(yline)[below] {\scriptsize$\beta$};

	\draw[dotted] (2,1) -- (2,-0.5) node(yline)[below] {\scriptsize$t_{\alpha}$};
	\draw[dotted] (5.5,1) -- (5.5,-0.5) node(yline)[below] {\scriptsize$t_{\beta+1}$};
	
	\draw[dotted] (0,1) -- (0,-0.5) node(yline)[below] {\scriptsize$0$};
	\draw[dotted] (8,1) -- (8,-0.5) node(yline)[below] {\scriptsize$P$};
	
	\draw[decoration={brace,raise=5pt},decorate]
	(2,.6) -- node[above=6pt] {\scriptsize $\rho_1$} (2.95,.6);
	\draw[decoration={brace,raise=5pt},decorate]
	(4.8,.6) -- node[above=6pt] {\scriptsize $\rho_2$} (5.5,.6);

	\draw[axis] (-.1,0)  -- (8.5,0);
	\end{tikzpicture}
	\caption{sequence $\sigma_{X_{\kappa,{\rho_1}},Y_{\kappa,{\rho_2}}}$ for tuple $(\kappa,\rho_1,\rho_2)$}
	\label{fig:sequencefortuple}
\end{figure}

We propose two dynamic programs to obtain ${X}_{\kappa,\rho}$ and ${Y}_{\kappa,\rho}$ for each pair $(\kappa, \rho)$. The first dynamic program (DP1) computes for fixed $\rho \le K^r$ the corresponding sets ${X}_{\kappa,\rho}$ for each choice of $\kappa$.
The DP is based on the following observation: Since, within each block of the sequence, jobs are ordered according to WSPT, the completion time $C_j$ of each job $j$ is determined entirely by the fact whether or not $j \in {X}_{\kappa,\rho}$ and by the total processing time $\varrho = p({X}_{\kappa,\rho} \cap J[\alpha, j])$ of jobs with index at most $j$ in ${X}_{\kappa,\rho}$. If $j \in {X}_{\kappa,\rho}$, then $C_j = t_{\alpha} + \varrho$ (see \Cautoref{fig:seqjDP1b}), otherwise $C_j = p(J[1,j]) +\rho - \varrho$ (see \Cautoref{fig:seqjDP1c}). Thus, iterating over the jobs in WSPT order, for each $j \in J[\alpha, \beta-1]$ and each $\varrho \leq \rho$, the DP constructs a set $X \subseteq J[\alpha, j]$ with $p(X) = \varrho$ so as to minimize the total weighted completion time of the jobs in $J[\alpha, j]$.

Formally, the DP considers states $(j, \varrho)$ with $j \in J[\alpha, \beta-1]$ and $\varrho \leq \rho$. We introduce a cost function $\theta_{1,\rho}(j,\varrho)$ which denotes the total weighted completion time of jobs sequenced so far (i.e., jobs in $J[\alpha,j]$). This cost function $\theta_{1,\rho}(j,\varrho)$ is computed recursively as follows:
\begin{align*}
\theta_{1,\rho}(\alpha-1,\varrho) = &\left\lbrace \begin{array}{ll}
0 & \text{ if } \varrho = 0 \\
\infty & \text{ otherwise}
\end{array} \right., \\
\theta_{1,\rho}(j,\varrho) = \min &\left\lbrace 
\begin{array}{l}
\left\lbrace \begin{array}{ll}
	\theta_{1,\rho}(j-1,\varrho-p_j) + w_j \cdot (t_{\alpha} + \varrho) & \text{if } j\in J^o\\
	\infty & \text{if } j\in J^r
\end{array}	\right. \\
	\theta_{1,\rho}(j-1,\varrho) + w_j \cdot \left(p(J[1,j]) + \rho - \varrho\right)
\end{array} \!\!\!\!\!\! \right\rbrace.
\end{align*}

This recursion runs in $\mathrm{O}(nP)$. We immediately see that $f_\beta(X_{\beta,\rho}) = \theta_{1,\rho}(\beta,\rho)$ and the corresponding set ${X}_{\beta,\rho}$ can be retrieved, in $\mathrm{O}(n)$ time, by traversing the state space backward starting from state $(\beta-1,\rho)$ and each time choosing the state leading to the minimum associated cost. Interestingly, as a byproduct of the above DP, we obtain ${X}_{\kappa,\rho}$ for all $\kappa$ with $\alpha < \kappa \leq \beta$ simply by traversing the state space backward starting from $(\kappa-1,\rho)$. This works since the cost values for states do not depend on $\kappa$. 
However, note that the cost function $\theta_{1,\rho}$ does depend on the target processing time $\rho$ for the jobs to be included in $X$.
Thus, we must run DP1 for each choice of $\rho \leq K^r$. Therefore, all subsets $X_{\kappa,\rho}$ are obtained in $\mathrm{O}(nP^2)$ time.

	\begin{figure}
	\centering
	\begin{subfigure}[b]{\linewidth}
		\centering
		\begin{tikzpicture}[
		axis/.style={thick, ->, >=stealth'},
		important line/.style={very thick},
		dashed line/.style={dashed,  thick},
		every node/.style={color=black,}
		]

		\node[fill=green!10, rectangle, draw=black, minimum height=.5cm,minimum width=2cm] at (1.5,0.5) {};
		\node[preaction={fill=green!10}, fill=green!10,  pattern=north west lines, pattern color=red!10, rectangle, draw=black, minimum height=.5cm,minimum width=2cm] at (4.5,0.5) {};
		
		\draw[dotted] (.5,1) -- (.5,-0.5) node(yline)[below] {\scriptsize$t_{\alpha}$};
		\draw[dotted] (3.5,1) -- (3.5,-0.5) node(yline)[below] {\scriptsize$t_{\alpha}+\rho$};
		\draw[dotted] (7.5,1) -- (7.5,-0.5) node(yline)[below] {\scriptsize$t_{\kappa}$};

		\draw[axis] (-.1,0)  -- (8.5,0);
		\end{tikzpicture}
		\caption{The situation before sequencing job $j$.}
		\label{fig:seqjDP1a}
	\end{subfigure}
	\vfill\vfill
	\begin{subfigure}[b]{\linewidth}
		\centering
		\begin{tikzpicture}[
		axis/.style={thick, ->, >=stealth'},
		important line/.style={very thick},
		dashed line/.style={dashed,  thick},
		every node/.style={color=black,}
		]

		\node[fill=green!10, rectangle, draw=black, minimum height=.5cm,minimum width=2cm] at (1.5,0.5) {};
		\node[preaction={fill=green!10}, fill=green!10,  pattern=north west lines, pattern color=red!10, rectangle, draw=black, minimum height=.5cm,minimum width=2cm] at (4.5,0.5) {};
		\node[fill=green!10, rectangle, draw=black, minimum height=.5cm,minimum width=.5cm] at (2.75,0.5) {\scriptsize$j$};

		\draw[dotted] (.5,1) -- (.5,-0.5) node(yline)[below] {\scriptsize$t_{\alpha}$};
		\draw[dotted] (3.5,1) -- (3.5,-0.5) node(yline)[below] {\scriptsize$t_{\alpha}+\rho$};
		\draw[dotted] (7.5,1) -- (7.5,-0.5) node(yline)[below] {\scriptsize$t_{\kappa}$};

		\draw[decoration={brace,raise=5pt},decorate]
		(.5,.7) -- node[above=6pt] {\scriptsize $\varrho$} (3,.7);

		\draw[axis] (-.1,0)  -- (8.5,0);
		\end{tikzpicture}
		\caption{Job $j$ is assigned to $X_{\kappa,\rho}$ and is completed before $t_{\alpha}+\rho$.}
		\label{fig:seqjDP1b}
	\end{subfigure}
	\vfill\vfill
	\begin{subfigure}[b]{\linewidth}
		\centering
		\begin{tikzpicture}[
		axis/.style={thick, ->, >=stealth'},
		important line/.style={very thick},
		dashed line/.style={dashed,  thick},
		every node/.style={color=black,}
		]

		\node[fill=green!10, rectangle, draw=black, minimum height=.5cm,minimum width=2cm] at (1.5,0.5) {};
		\node[preaction={fill=green!10}, fill=green!10,  pattern=north west lines, pattern color=red!10, rectangle, draw=black, minimum height=.5cm,minimum width=2cm] at (4.5,0.5) {};
		\node[fill=green!10, rectangle, draw=black, minimum height=.5cm,minimum width=.5cm] at (5.75,0.5) {\scriptsize$j$};
		
		\draw[dotted] (.5,1) -- (.5,-0.5) node(yline)[below] {\scriptsize$t_{\alpha}$};
		\draw[dotted] (3.5,1) -- (3.5,-0.5) node(yline)[below] {\scriptsize$t_{\alpha}+\rho$};
		\draw[dotted] (7.5,1) -- (7.5,-0.5) node(yline)[below] {\scriptsize$t_{\kappa}$};
		\draw[dotted] (6,1) -- (6,-0.5) node(yline)[below] {\scriptsize$p(J[1,j]) + \rho - \varrho$};
		
		\draw[decoration={brace,raise=5pt},decorate]
		(2.5,.7) -- node[above=6pt] {\scriptsize $\rho-\varrho$} (3.5,.7);

		\draw[axis] (-.1,0)  -- (8.5,0);
		\end{tikzpicture}
		\caption{Job $j$ is assigned to $\bar{X}_{\kappa,\rho}$ and is completed after $t_{\alpha}+\rho$.}
		\label{fig:seqjDP1c}
	\end{subfigure}
	
	\caption{Deciding on the position of job $j \in H$ in DP1 and DP5}
	\label{fig:seqjDP1}
\end{figure}
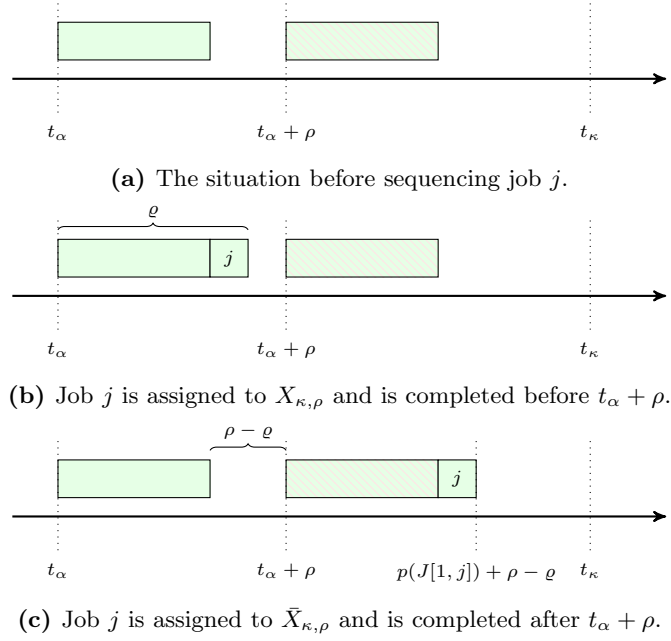

By a symmetric argument we can design DP2 to compute $Y_{\kappa,\rho}$ for all $\kappa$ and $\rho$ in time $\mathrm{O}(nP^2)$. \Cautoref{fig:seqjDP2a} to \Cautoref{fig:seqjDP2c} support the intuition about how completion time of job $j$ is determined by the fact whether or not $j \in {Y}_{\kappa,\rho}$ and by the total processing time $\varrho = p({Y}_{\kappa,\rho} \cap J[j, \beta])$.

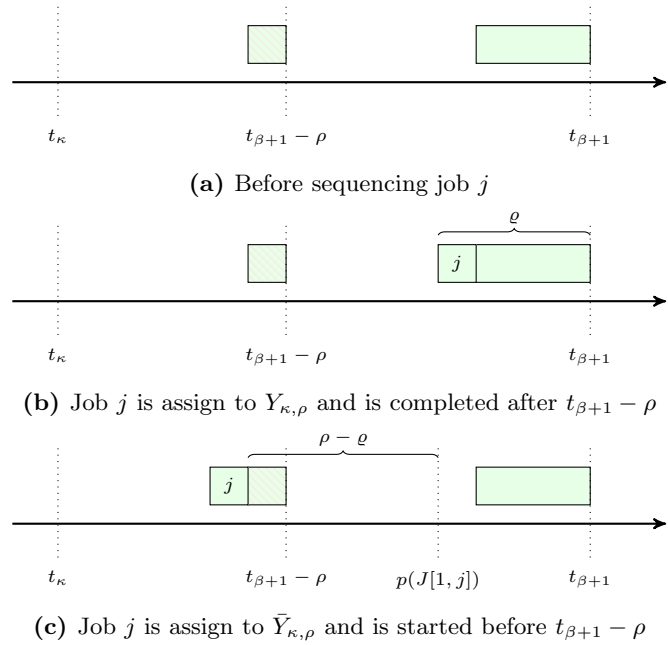
\begin{figure}
	\centering
	\begin{subfigure}[b]{\linewidth}
		\centering
		\begin{tikzpicture}[
		axis/.style={thick, ->, >=stealth'},
		important line/.style={very thick},
		dashed line/.style={dashed,  thick},
		every node/.style={color=black,}
		]

		\node[preaction={fill=green!10}, fill=green!10,  pattern=north west lines, pattern color=red!10, rectangle, draw=black, minimum height=.5cm,minimum width=.5cm] at (3.25,0.5) {};
		\node[fill=green!10,  rectangle, draw=black, minimum height=.5cm,minimum width=1.5cm] at (6.75,0.5) {};

		\draw[dotted] (.5,1) -- (.5,-0.5) node(yline)[below] {\scriptsize$t_{\kappa}$};
		\draw[dotted] (3.5,1) -- (3.5,-0.5) node(yline)[below] {\scriptsize$t_{\beta+1}-\rho$};
		\draw[dotted] (7.5,1) -- (7.5,-0.5) node(yline)[below] {\scriptsize$t_{\beta+1}$};

		\draw[axis] (-.1,0)  -- (8.5,0);
		\end{tikzpicture}
		\caption{Before sequencing job $j$}
		\label{fig:seqjDP2a}
	\end{subfigure}
	\vfill\vfill
\begin{subfigure}[b]{\linewidth}
	\centering
	\begin{tikzpicture}[
	axis/.style={thick, ->, >=stealth'},
	important line/.style={very thick},
	dashed line/.style={dashed,  thick},
	every node/.style={color=black,}
	]

	\node[preaction={fill=green!10}, fill=green!10,  pattern=north west lines, pattern color=red!10, rectangle, draw=black, minimum height=.5cm,minimum width=.5cm] at (3.25,0.5) {};
	\node[fill=green!10,  rectangle, draw=black, minimum height=.5cm,minimum width=1.5cm] at (6.75,0.5) {};
	\node[fill=green!10, rectangle, draw=black, minimum height=.5cm,minimum width=.5cm] at (5.75,0.5) {\scriptsize$j$};

	\draw[dotted] (.5,1) -- (.5,-0.5) node(yline)[below] {\scriptsize$t_{\kappa}$};
	\draw[dotted] (3.5,1) -- (3.5,-0.5) node(yline)[below] {\scriptsize$t_{\beta+1}-\rho$};
	\draw[dotted] (7.5,1) -- (7.5,-0.5) node(yline)[below] {\scriptsize$t_{\beta+1}$};
	
	\draw[decoration={brace,raise=5pt},decorate]
	(5.5,.7) -- node[above=6pt] {\scriptsize $\varrho$} (7.5,.7);

	\draw[axis] (-.1,0)  -- (8.5,0);
	\end{tikzpicture}
	\caption{Job $j$ is assign to $Y_{\kappa,\rho}$ and is completed after $t_{\beta+1}-\rho$}
	\label{fig:seqjDP2b}
\end{subfigure}
	\vfill\vfill
	\begin{subfigure}[b]{\linewidth}
		\centering
		\begin{tikzpicture}[
		axis/.style={thick, ->, >=stealth'},
		important line/.style={very thick},
		dashed line/.style={dashed,  thick},
		every node/.style={color=black,}
		]

	\node[preaction={fill=green!10}, fill=green!10,  pattern=north west lines, pattern color=red!10, rectangle, draw=black, minimum height=.5cm,minimum width=.5cm] at (3.25,0.5) {};
	\node[fill=green!10,  rectangle, draw=black, minimum height=.5cm,minimum width=1.5cm] at (6.75,0.5) {};
		\node[fill=green!10, rectangle, draw=black, minimum height=.5cm,minimum width=.5cm] at (2.75,0.5) {\scriptsize$j$};

		\draw[dotted] (.5,1) -- (.5,-0.5) node(yline)[below] {\scriptsize$t_{\kappa}$};
		\draw[dotted] (3.5,1) -- (3.5,-0.5) node(yline)[below] {\scriptsize$t_{\beta+1}-\rho$};
		\draw[dotted] (7.5,1) -- (7.5,-0.5) node(yline)[below] {\scriptsize$t_{\beta+1}$};
		\draw[dotted] (5.5,1) -- (5.5,-0.5) node(yline)[below] {\scriptsize$p(J[1,j])$};
		
		\draw[decoration={brace,raise=5pt},decorate]
		(3,.7) -- node[above=6pt] {\scriptsize $\rho-\varrho$} (5.5,.7);

		\draw[axis] (-.1,0)  -- (8.5,0);
		\end{tikzpicture}
		\caption{Job $j$ is assign to $\bar{Y}_{\kappa,\rho}$ and is started before $t_{\beta+1}-\rho$}
		\label{fig:seqjDP2c}
	\end{subfigure}

	\caption{Deciding on the position of job $j \in H$ in DP2 and DP6}
	\label{fig:seqjDP2}
\end{figure}

Finally, we show that searching over all $(\kappa,\rho_1,\rho_2) \in \Xi$ to find $X^*$ and $Y^*$ can be done in $\mathrm{O}(nP)$ time. We say $X_{\kappa,\rho}$ dominates $X_{\kappa,\rho'}$ if $\rho > \rho'$ and $f(X_{\kappa,\rho}) \le f(X_{\kappa,\rho'})$ and $Y_{\kappa,\rho}$ dominates $Y_{\kappa,\rho'}$ if $\rho > \rho'$ and $g(Y_{\kappa,\rho}) \le g(Y_{\kappa,\rho'})$. For each $\kappa$, we compile a set $\mathcal{X}_\kappa'$ of non-dominated sets $X_{\kappa,\rho}$ and a set $\mathcal{Y}_\kappa'$ of non-dominated sets $Y_{\kappa,\rho}$, both of which are sorted in decreasing order of $\rho$. Then for each $\kappa$, we scan through $\mathcal{X}_\kappa'$, each time choose $X_{\kappa,\rho} \in \mathcal{X}_\kappa'$ and only pair it with $Y_{\kappa,\rho'} \in \mathcal{Y}_\kappa'$ with 
$$ \rho' = \min \left\lbrace \bar{\rho} \mid Y_{\kappa,\bar{\rho}} \in \mathcal{Y}_\kappa', p(J[\alpha,\beta]) - \rho - \bar{\rho} \le K^r \right\rbrace. $$ Among the pairs, we choose the one which minimizes $f(X_{\kappa,\rho})+g(Y_{\kappa,\rho'})$. Generating and scanning through the dominating sets both are done in $\mathrm{O}(nP)$ time.
\end{proof}

\begin{lemma}
	\label{lem:twc_nPW}
	$1|er|\sum w_jC_j$ can be solved in $\mathrm{O}(nPW)$-time.
\end{lemma}

\begin{proof}
	We define ${\cal{X}}_{\kappa,\rho},{\cal{Y}}_{\kappa,\rho},f_{\kappa}(X), g_{\kappa}(Y)$ and compute $X_{\kappa,\rho},Y_{\kappa,\rho},\bar{X}_{\kappa,\rho},\bar{Y}_{\kappa,\rho}, X^*$ and $Y^*$ similarly to the proof of \Cautoref{lem:twc_nP2}. 
	
We propose two DPs to obtain ${X}_{\kappa,\rho}$ and ${Y}_{\kappa,\rho}$. The first DP (DP3) computes the corresponding sets ${X}_{\kappa,\rho}$ for each choice of $\kappa$ and for each choice of $\rho$.
	In DP3, we use states $(j,\varrho,\omega)$ that stores the current $j$, the total processing time of jobs added to $X_{\kappa,\rho}$ so far, and the total weight of jobs in $\bar{X}_{\kappa,\rho}$. We check jobs in $J[\alpha,\kappa-1]$ one by one in WSPT order and decide whether to add job $j \in H$ to $X_{\kappa,\rho}$ or not (see \Cautoref{fig:seqjDP3a}). If we decide to add job $j$ to $X_{\kappa,\rho}$, then $C_j = t_{\alpha} + \varrho$ (see \Cautoref{fig:seqjDP3b}), otherwise job $j$ is temporarily set to be completed at $p(J[1,j])$ (see \Cautoref{fig:seqjDP3c}) but could be shifted to the right if more jobs are to be added to $X_{\kappa,\rho}$. The extent of such a shift depends on the jobs in $J[j+1,\kappa-1]$ that will be eventually added to $X_{\kappa,\rho}$. However, since such information is not available at state $(j,\varrho,\omega)$, when adding job $j$ to $\bar{X}_{\kappa,\rho}$, we only consider its temporary completion time while computing its cost $w_j \cdot p(J[1,j])$ and later when more information is available, we add extra costs: whenever a job $j$ is added to ${X}_{\kappa,\rho}$, for which a cost of $w_j \cdot (t_{\alpha}+\varrho)$ is incurred, jobs in $\bar{X}_{\kappa,\rho}$ also move $p_j$ time units to the right that induces an extra cost of $\omega p_j$ (recall that $\omega$ is the weight of jobs added to $\bar{X}_{\kappa,\rho}$ so far). We introduce a cost function $\theta_2(j,\varrho,\omega)$ which is the total weighted completion time of jobs sequenced so far (i.e., $J[\alpha,j]$). This cost function 
	is computed recursively as follows:
	\begin{align*}
	&\theta_2(\alpha-1,\varrho,\omega) = \left\lbrace \begin{array}{ll}
	0 & \text{ if } \varrho = 0, \omega = 0 \\
	\infty & \text{ otherwise}
	\end{array} \right., \\
	&\theta_2(j,\varrho,\omega) = \\
	&\min \left\lbrace 
	\begin{array}{l}
	\left\lbrace \begin{array}{ll}
	\theta_2(j-1,\varrho-p_j,\omega) + w_j(t_{\alpha} + \varrho)+\omega p_j & \text{if } j\in J^o\\
	\infty & \text{if } j\in J^r
	\end{array}	\right. \\
	\theta_2(j-1,\varrho,\omega-w_j) + w_j p(J[1,j])
	\end{array} \!\!\!\!\!\! \right\rbrace.
	\end{align*}

	This recursion runs in $\mathrm{O}(nPW)$ time. 
	We see that $f_\beta(X_{\beta,K^r}) = \theta_{2}(\beta,K^r)$ and the corresponding set ${X}_{\beta,K^r}$ can be retrieved, in $\mathrm{O}(n)$ time, by traversing  the state space backward starting from $(\beta-1,K^r,w^*)$ with $$w^* := \argmin_{w \in [0,w(J[\alpha,\beta])]} \left\lbrace \theta_2(\beta-1,K^r,w)\right\rbrace,$$ each time choosing the state with minimum cost. Interestingly, as a byproduct of the above DP, we obtain ${X}_{\kappa,\rho}$ for all $\kappa$ with $\alpha < \kappa \leq \beta$ and all $\rho$ with $0 < \rho \le K^r$ simply by traversing the state space backward starting from $(\kappa-1,\rho)$. This works since the cost values for states do not depend on $\kappa$ and $\rho$. Therefore, all subsets ${X}_{\kappa,\rho}$ combined are obtained in $\mathrm{O}(nPW)$ time.
	
	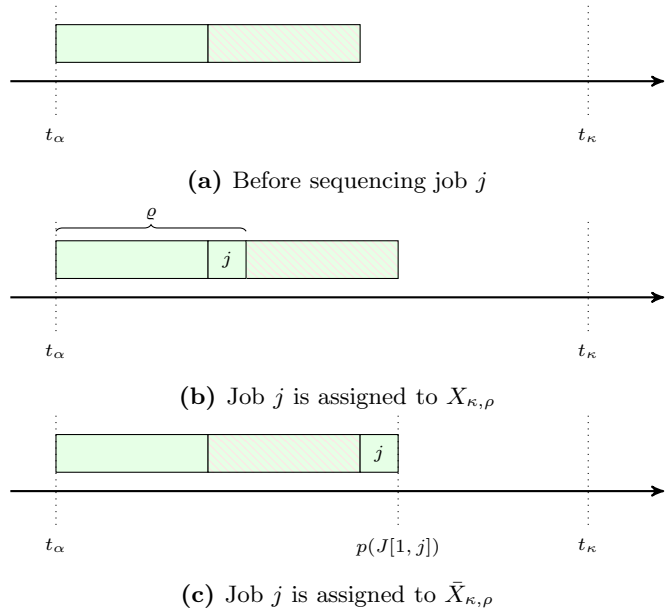
\begin{figure}
		\centering
		\begin{subfigure}[b]{\linewidth}
			\centering
			\begin{tikzpicture}[
			axis/.style={thick, ->, >=stealth'},
			important line/.style={very thick},
			dashed line/.style={dashed,  thick},
			every node/.style={color=black,}
			]

			\node[fill=green!10, rectangle, draw=black, minimum height=.5cm,minimum width=2cm] at (1.5,0.5) {};
			\node[preaction={fill=green!10}, fill=green!10,  pattern=north west lines, pattern color=red!10, rectangle, draw=black, minimum height=.5cm,minimum width=2cm] at (3.5,0.5) {};
			
			\draw[dotted] (.5,1) -- (.5,-0.5) node(yline)[below] {\scriptsize$t_{\alpha}$};
			\draw[dotted] (7.5,1) -- (7.5,-0.5) node(yline)[below] {\scriptsize$t_{\kappa}$};

			\draw[axis] (-.1,0)  -- (8.5,0);
			\end{tikzpicture}
			\caption{Before sequencing job $j$}
			\label{fig:seqjDP3a}
		\end{subfigure}
		\vfill\vfill
		\begin{subfigure}[b]{\linewidth}
			\centering
			\begin{tikzpicture}[
			axis/.style={thick, ->, >=stealth'},
			important line/.style={very thick},
			dashed line/.style={dashed,  thick},
			every node/.style={color=black,}
			]

			\node[fill=green!10, rectangle, draw=black, minimum height=.5cm,minimum width=2cm] at (1.5,0.5) {};
			\node[preaction={fill=green!10}, fill=green!10,  pattern=north west lines, pattern color=red!10, rectangle, draw=black, minimum height=.5cm,minimum width=2cm] at (4,0.5) {};
			\node[fill=green!10, rectangle, draw=black, minimum height=.5cm,minimum width=.5cm] at (2.75,0.5) {\scriptsize$j$};

			\draw[dotted] (.5,1) -- (.5,-0.5) node(yline)[below] {\scriptsize$t_{\alpha}$};
			\draw[dotted] (7.5,1) -- (7.5,-0.5) node(yline)[below] {\scriptsize$t_{\kappa}$};

			\draw[decoration={brace,raise=5pt},decorate]
			(.5,.7) -- node[above=6pt] {\scriptsize $\varrho$} (3,.7);

			\draw[axis] (-.1,0)  -- (8.5,0);
			\end{tikzpicture}
			\caption{Job $j$ is assigned to $X_{\kappa,\rho}$}
			\label{fig:seqjDP3b}
		\end{subfigure}
		\vfill\vfill
		\begin{subfigure}[b]{\linewidth}
			\centering
			\begin{tikzpicture}[
			axis/.style={thick, ->, >=stealth'},
			important line/.style={very thick},
			dashed line/.style={dashed,  thick},
			every node/.style={color=black,}
			]

			\node[fill=green!10, rectangle, draw=black, minimum height=.5cm,minimum width=2cm] at (1.5,0.5) {};
			\node[preaction={fill=green!10}, fill=green!10,  pattern=north west lines, pattern color=red!10, rectangle, draw=black, minimum height=.5cm,minimum width=2cm] at (3.5,0.5) {};
			\node[fill=green!10, rectangle, draw=black, minimum height=.5cm,minimum width=.5cm] at (4.75,0.5) {\scriptsize$j$};
			
			\draw[dotted] (.5,1) -- (.5,-0.5) node(yline)[below] {\scriptsize$t_{\alpha}$};
			\draw[dotted] (7.5,1) -- (7.5,-0.5) node(yline)[below] {\scriptsize$t_{\kappa}$};
			\draw[dotted] (5,1) -- (5,-0.5) node(yline)[below] {\scriptsize$p(J[1,j])$};

			\draw[axis] (-.1,0)  -- (8.5,0);
			\end{tikzpicture}
			\caption{Job $j$ is assigned to $\bar{X}_{\kappa,\rho}$}
			\label{fig:seqjDP3c}
		\end{subfigure}
		
		\caption{Deciding on the position of job $j \in H$ in DP3}
		\label{fig:seqjDP3}
	\end{figure}

By a symmetric argument we can design DP4 to compute $Y_{\kappa,\rho}$ for each $\alpha \leq \kappa < \beta$ and $1\le\rho<K^r$ in time $\mathrm{O}(nPW)$. \Cautoref{fig:seqjDP4a} to \Cautoref{fig:seqjDP4c} support the intuition about how completion time of job $j$ is determined.

	\begin{figure}
		\centering
		\begin{subfigure}[b]{\linewidth}
			\centering
			\begin{tikzpicture}[
			axis/.style={thick, ->, >=stealth'},
			important line/.style={very thick},
			dashed line/.style={dashed,  thick},
			every node/.style={color=black,}
			]

			\node[fill=green!10, rectangle, draw=black, minimum height=.5cm,minimum width=2cm] at (6.5,0.5) {};
			\node[preaction={fill=green!10}, fill=green!10,  pattern=north west lines, pattern color=red!10, rectangle, draw=black, minimum height=.5cm,minimum width=2cm] at (4.5,0.5) {};
			
			\draw[dotted] (.5,1) -- (.5,-0.5) node(yline)[below] {\scriptsize$t_{\kappa}$};
			\draw[dotted] (7.5,1) -- (7.5,-0.5) node(yline)[below] {\scriptsize$t_{\beta+1}$};

			\draw[axis] (-.1,0)  -- (8.5,0);
			\end{tikzpicture}
			\caption{Before sequencing job $j$}
			\label{fig:seqjDP4a}
		\end{subfigure}
		\vfill\vfill
		\begin{subfigure}[b]{\linewidth}
			\centering
			\begin{tikzpicture}[
			axis/.style={thick, ->, >=stealth'},
			important line/.style={very thick},
			dashed line/.style={dashed,  thick},
			every node/.style={color=black,}
			]

			\node[fill=green!10, rectangle, draw=black, minimum height=.5cm,minimum width=2cm] at (6.5,0.5) {};
			\node[preaction={fill=green!10}, fill=green!10,  pattern=north west lines, pattern color=red!10, rectangle, draw=black, minimum height=.5cm,minimum width=2cm] at (4,0.5) {};
			\node[fill=green!10, rectangle, draw=black, minimum height=.5cm,minimum width=.5cm] at (5.25,0.5) {\scriptsize$j$};

			\draw[dotted] (.5,1) -- (.5,-0.5) node(yline)[below] {\scriptsize$t_{\kappa}$};
			\draw[dotted] (7.5,1) -- (7.5,-0.5) node(yline)[below] {\scriptsize$t_{\beta+1}$};

			\draw[decoration={brace,raise=5pt},decorate]
			(5,.7) -- node[above=6pt] {\scriptsize $\varrho$} (7.5,.7);

			\draw[axis] (-.1,0)  -- (8.5,0);
			\end{tikzpicture}
			\caption{Job $j$ is assigned to $Y_{\kappa,\rho}$}
			\label{fig:seqjDP4b}
		\end{subfigure}
		\vfill\vfill
		\begin{subfigure}[b]{\linewidth}
			\centering
			\begin{tikzpicture}[
			axis/.style={thick, ->, >=stealth'},
			important line/.style={very thick},
			dashed line/.style={dashed,  thick},
			every node/.style={color=black,}
			]

			\node[fill=green!10, rectangle, draw=black, minimum height=.5cm,minimum width=2cm] at (6.5,0.5) {};
			\node[preaction={fill=green!10}, fill=green!10,  pattern=north west lines, pattern color=red!10, rectangle, draw=black, minimum height=.5cm,minimum width=2cm] at (4.5,0.5) {};
			\node[fill=green!10, rectangle, draw=black, minimum height=.5cm,minimum width=.5cm] at (3.25,0.5) {\scriptsize$j$};
			
			\draw[dotted] (.5,1) -- (.5,-0.5) node(yline)[below] {\scriptsize$t_{\kappa}$};
			\draw[dotted] (7.5,1) -- (7.5,-0.5) node(yline)[below] {\scriptsize$t_{\beta+1}$};
			\draw[dotted] (3.5,1) -- (3.5,-0.5) node(yline)[below] {\scriptsize$p(J[1,j])$};

			\draw[axis] (-.1,0)  -- (8.5,0);
			\end{tikzpicture}
			\caption{Job $j$ is assigned to $\bar{Y}_{\kappa,\rho}$}
			\label{fig:seqjDP4c}
		\end{subfigure}
		
		\caption{Deciding on the position of job $j \in H$ in DP4}
		\label{fig:seqjDP4}
	\end{figure}
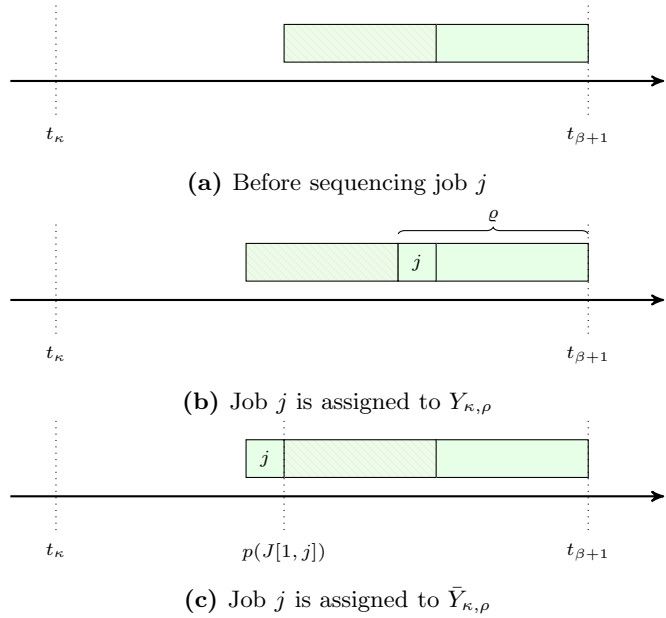

Finally, we argue that searching over all $(\kappa,\rho_1,\rho_2) \in \Xi$ to find $X^*$ and $Y^*$ can be done in $\mathrm{O}(nP)$ (see the final paragraph in the proof of \Cautoref{lem:twc_nP2}), the proof is concluded.
\end{proof}

From \Cautoref{lem:twc_nP2} and \Cautoref{lem:twc_nPW} we infer the following theorem.

\begin{theorem}
	\label{totalwc}
	$1|er|\sum w_jC_j$ can be solved in $\mathrm{O}(nP \min\{P,W\})$-time.
\end{theorem}

The following corollary is immediate. 
\begin{corollary}
	$1|er|\sum C_j$ can be solved in $\mathrm{O}(n^2P)$-time.
\end{corollary}	

\subsection{Maximum lateness}
We first show the NP-hardness of $1|er|\max L_j$ and then we propose a pseudo-polynomial time approach to solve $1|er|\max L_j$. 
\begin{theorem}
	\label{thm:maxlateness}
	$1|er|\max L_j$ is NP-hard.
\end{theorem}

\begin{proof}
We prove the NP-hardness of $1|er|\max L_j$ by a reduction from \textsc{Partition} which is known to be NP-hard, see \citep{Garey1979}.
	
\textsc{Partition}: Given integer numbers $a_1,\ldots,a_m$, is there a subset of $\{a_1,\ldots,a_m\}$ with total value of $B = \frac{1}{2}\sum_{i=1}^{m} a_i$?
	
Given an instance of \textsc{Partition}, we construct an instance of $1|er|\max L_j$ with $m+2$ jobs as follows:
\begin{itemize}
	\item $J=\{1,\ldots,m+2\}$ and $J^r=\{m+1,m+2\}$,
	\item $p_j=a_j$ and $d_j=2B+1$ for each $j=1,\ldots,m$,
	\item $p_{m+1}=1$, $d_{m+1}=B+1$, $p_{m+2}=1$ and $d_{m+2}=2B+2$, and
	\item $K^r=B+2$.
\end{itemize}  
	
We claim that there is a feasible schedule with maximum lateness of at most zero if and only if the answer to the instance of \textsc{Partition} is yes. Notice that zero is also a lower bound to maximum lateness since no due date exceeds the makespan of $2B+2$. 
	
Let us consider a schedule with lateness zero. Job $m+2$ is scheduled last since  it is the only job with due date $2B+2$. Job $m+1$ is started exactly at $B$ since it cannot be started before $B$ due to feasibility and it cannot be started after $B$ without being tardy. Hence, we conclude that the subsets of jobs before and after job $m+1$ both have a total processing time of $B$ and, thus, constitute a yes-certificate for the corresponding instance of \textsc{Partition}. \Cautoref{fig:zerolateness} depicts the structure of the schedule.

Second, if a yes-certificate for the instance of \textsc{Partition} is given we can construct a sequence with the structure discussed above and, thus, yielding maximum lateness of at most zero. This completes the proof.
\end{proof}

	\begin{figure*}
		\centering
		\resizebox{\linewidth}{!}{%
		\begin{tikzpicture}[
		axis/.style={thick, ->, >=stealth'},
		important line/.style={very thick},
		dashed line/.style={dashed,  thick},
		every node/.style={color=black,}
		]

		
		\node[fill=red!10, rectangle, draw=black, minimum height=.5cm,minimum width=1cm] at (13.5,0.5) {\scriptsize$m+2$};
		\node[fill=red!10, rectangle, draw=black, minimum height=.5cm,minimum width=1cm] at (6.5,0.5) {\scriptsize$m+1$};
		\node[fill=green!10, rectangle, draw=black, minimum height=.5cm,minimum width=6cm, label=] at (10,0.5) {\scriptsize Jobs after $m+1$};
		\node[fill=green!10, rectangle, draw=black, minimum height=.5cm,minimum width=6cm, label=] at (3,0.5) {\scriptsize Jobs before $m+1$};
		\draw[dotted] (0,1) -- (0,-0.5) node(yline)[below] {\scriptsize$0$};
		\draw[dotted] (7,1) -- (7,-0.5) node(yline)[below] {\scriptsize$d_{m+1} = B+1$};
		\draw[dotted] (13,1) -- (13,-0.1) node(yline)[below] {\scriptsize$d_j =2B+1$};
		\draw[dotted] (14,1) -- (14,-0.5) node(yline)[below] {\scriptsize$d_{m+2} =2B+2$};
		\draw[decoration={brace,raise=5pt},decorate]
		(0,1) -- node[above=6pt] {\scriptsize  $B$} (6,1);
		\draw[decoration={brace,raise=5pt},decorate]
		(7,1) -- node[above=6pt] {\scriptsize $B$} (13,1);
		\draw[axis] (-.1,0)  -- (14.5,0);
		\end{tikzpicture}
	}
		\caption{The schedule with zero lateness as described in theorem \Cautoref{thm:maxlateness}}
		\label{fig:zerolateness}
	\end{figure*}
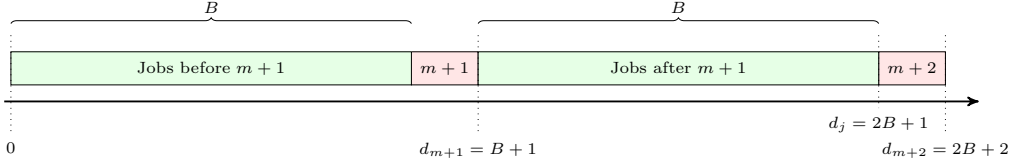

We now show that $1|er|\max L_j$ can be solved in pseudo-polynomial time. 
First, we will show that there always exists an optimal sequence with a special structure consisting of five blocks that are internally ordered according to the \emph{earliest due date} (EDD) rule. 

Without loss of generality, we assume the jobs to be numbered according to EDD (i.e., $d_1 \le \dots \le d_n$). We also let $\alpha = \min J^r$, $\beta = \max J^r$, and $H = J[\alpha, \beta] \setminus J^r$. Let $X, Y \subseteq H$ with $\max X < \min Y$. We construct a sequence $\sigma_{X, Y}$ as outlined in \Cautoref{subsec:twc}, but now with the jobs within each block being ordered according to EDD.

\begin{lemma}\label{lem:lmax-opt-structure}
For each instance of $1|er|\max L_j$ there exists $X^*, Y^* \subseteq H$ with $\max X^* < \min Y^*$ such that $\sigma_{X^*, Y^*}$ is optimal.
\end{lemma}

\begin{proof}
	Let $\sigma$ be an optimal sequence. Let $i_1 := \min \{i : \sigma(i) \in J^r\}$, $i_2 := \max \{i : \sigma(i) \in J^r\}$  and $S = \{\sigma(i_1), \dots, \sigma(i_2)\}$. Similarly to the proof of \Cautoref{lem:twc-opt-structure}, we see that $\sum_{j \in S} p_j \leq K^r$ by feasibility of $\sigma$ and any sequence that schedules jobs in $S$ consecutively is feasible. Therefore, rearranging the jobs within $S$ according to EDD maintains feasibility of $\sigma$ without increasing maximum lateness. We can thus assume, without loss of generality,~$\alpha = \sigma(i_1) < \dots < \sigma(i_2) = \beta$ (i.e., the jobs in $S$ are scheduled according to EDD and, in particular, $S \subseteq J[\alpha, \beta]$).
	
Now consider the job set $J' := J \setminus S \cup \{j'\}$ where jobs in $S$ are merged into a single job $j'$ with due date 
	$$d_{j'} = \min \left\lbrace d_{\sigma(i_1+k)}+\sum_{i=i_1+k+1}^{i_2} p_{\sigma(i)}\mid k=0,\ldots,i_2-i_1 \right\rbrace.$$ 
Thus, the lateness of job $j'$ captures the maximum lateness among jobs $\sigma(i_1),\ldots,\sigma(i_2)$ if they are scheduled consecutively in EDD. Furthermore, we let $$X^* := \{j \in H \setminus S : d_j \le d_{j'} \}$$ and $$Y^* := H \setminus (S \cup X^*).$$ 
	
Note that $\max X^* < \min Y^*$ by construction and that $\sigma_{X^*, Y^*}$ is a feasible sequence for $J$. Further, note that both $\sigma_{X^*, Y^*}$ and $\sigma$ induce sequences $\sigma'_{X^*, Y^*}$ and $\sigma'$ for job set $J'$, respectively. In particular, $\sigma'_{X^*, Y^*}$ orders jobs in $J'$ according to EDD and, therefore,
\begin{align*}
\underbrace{L_{\max}(\sigma_{X^*, Y^*})  = L_{\max}(\sigma'_{X^*, Y^*})}_{L_{j'}(\sigma'_{X^*, Y^*}) = \max_{j\in S} L_j(\sigma_{X^*, Y^*})} \leq \underbrace{L_{\max}(\sigma') = L_{\max}(\sigma)}_{L_{j'}(\sigma') = \max_{j\in S} L_j(\sigma)},
\end{align*}
which shows the optimality of $\sigma_{X^*, Y^*}$.
\end{proof}

\begin{theorem}
	\label{thm:lmax_nP2}
	$1|er|\max L_j$ can be solved in $\mathrm{O}(nP)$-time.
\end{theorem}

\begin{proof}
Based on \Cautoref{lem:lmax-opt-structure}, it suffices to find $X^*, Y^* \subseteq H$ with $\max X^* < \min Y^*$ such that $\sigma_{X^*, Y^*}$ is optimal.
As before, we define $$\mathcal{X}_{\kappa,\rho} = \{X \subseteq H : \max X < \kappa, p(X) = \rho \}\text{ and}$$ $$\mathcal{Y}_{\kappa,\rho} = \{Y \subseteq H : \min Y \ge \kappa, p(Y) = \rho \}.$$ Furthermore, we define
	\begin{align*}
	f'_\kappa(X) & = \max_{j \in J[\alpha,\kappa - 1]} \left\lbrace C^{\sigma_{X,\emptyset}}_{j} 	- d_j \right\rbrace \quad \text{ and } \quad
	g'_\kappa(Y) = \max_{j \in J[\kappa, \beta]} \left\lbrace C^{\sigma_{\emptyset,Y}}_{j} - d_j \right\rbrace 
	\end{align*}
	and let $$X_{\kappa, \rho} \in \argmin_{X \in \mathcal{X}_{\kappa,\rho}} f'_{\kappa}(X) \quad \text{and} \quad Y_{\kappa, \rho} \in \argmin_{Y \in \mathcal{Y}_{\kappa,\rho}} g'_{\kappa}(Y)$$ 
	for each $\kappa \in J[\alpha, \beta]$ and $\rho \leq P$.
	
	Let $\theta_3(\kappa, \rho) = \min_{X \in \mathcal{X}_{\kappa, \rho}} f'_\kappa(X)$ and $\theta_4(\kappa, \rho) = \min_{Y \in \mathcal{Y}_{\kappa, \rho}} g'_\kappa(Y)$. We show that $\theta_3$ and $\theta_4$ can again be expressed by simple recursions, giving ways to dynamic programs for computing $X_{\kappa, \rho}$ and $Y_{\kappa, \rho}$ for all pairs $(\kappa, \rho)$, respectively. Once $X_{\kappa, \rho}$ and $Y_{\kappa, \rho}$ for all pairs $(\kappa, \rho)$ are computed, $\sigma_{X^*, Y^*}$ is obtained as $X^* = X_{\kappa^*,\rho_1^*}$ and $Y^* = {Y}_{\kappa^*,\rho_2^*}$, where 
	\begin{align*}
	(\kappa^*,\rho_1^*,\rho_2^*) \in \argmin_{(\kappa,\rho_1,\rho_2) \in \Xi} \left \lbrace \max \{f'_\kappa(X_{\kappa,\rho_1}), g'_\kappa(Y_{\kappa,\rho_2})\} \right\rbrace
	\end{align*}
	and $$\Xi = \{(\kappa,\rho_1,\rho_2) \; | \; {\cal{X}}_{\kappa,\rho_{1}},{\cal{Y}}_{\kappa,\rho_2} \neq \emptyset, p(J[\alpha,\beta]) - \rho_1 - \rho_2 \le K^r \}.$$
	The sequence $\sigma_{X^*, Y^*}$ minimizes maximum lateness among all such sequences because $C^{\sigma_{X^*, Y^*}}_j = C^{\sigma_{X^*, \emptyset}}_j$ for all $j \in J[\alpha, \kappa^* - 1]$, $C^{\sigma_{X^*, Y^*}}_j = C^{\sigma_{\emptyset, Y^*}}_j$ for all $j \in J[\kappa^*, \beta]$, and the completion time of all jobs $j \in J \setminus J[\alpha, \beta]$ is independent from the choice of $X^*$ and $Y^*$. Note that $\kappa^*,\rho_1^*,\rho_2^*$ can again be determined in time $\mathrm{O}(nP)$ (see the final paragraph in the proof of \Cautoref{lem:twc_nP2}).

  In order to obtain the recursion of $\theta_3$, we first prove the following two claims on the structure of the function $f'_\kappa$.	
	\begin{claim}
	For all $X \subseteq H$ and $\kappa \in J[\alpha, \beta]$, we have
	  $$f'_\kappa(X) = \max_{j \in J[\alpha,\kappa - 1] \setminus X} \left\lbrace C^{\sigma_{X,\emptyset}}_{j} - d_j \right\rbrace.$$ 
  \end{claim} 
  \begin{proof}
    Note that for all $j \in X$, we have $C^{\sigma_{X,\emptyset}}_{j} \leq C^{\sigma_{X,\emptyset}}_{\alpha}$ and $d_\alpha \leq d_j$. Thus $C^{\sigma_{X,\emptyset}}_{\alpha} - d_\alpha  \ge C^{\sigma_{X,\emptyset}}_{j} - d_j $ for any $j \in X$.
  \end{proof}
  \begin{claim}\label{claim:structure-maxl-X}
	  Let $X \subseteq H$ and $\kappa > \max X$. Then $f'_{\kappa+1}(X \cup \{\kappa\}) = f'_{\kappa}(X) + p_\kappa$ and $f'_{\kappa+1}(X) = \max \{f'_\kappa(X),\; p(J[1, \kappa]) - d_\kappa\}$.
	\end{claim}
	\begin{proof}
	  The first equality follows immediately from the preceding claim because $C^{\sigma_{X \cup \{\kappa\},\emptyset}}_j = C^{\sigma_{X,\emptyset}}_j + p_\kappa$ for all $j \in J[\alpha, \kappa-1] \setminus X$. The second equality follows from the definition of $f'_\kappa$ and the fact that $\kappa > \max X$, and hence job $\kappa$ precedes each job $j > \kappa$ in $\sigma_{X, \emptyset}$.
	\end{proof}
  From \Cautoref{claim:structure-maxl-X} we can deduce the following recursion for $\theta_3$:
	\begin{align*}
	\theta_3(\alpha + 1, \rho) = &\left\lbrace \begin{array}{ll}
	p(J[1, \alpha]) - d_\alpha & \text{ if } \rho = 0 \\
	\infty & \text{ otherwise}
	\end{array} \right., \\
	\theta_3(\kappa + 1, \rho) = \min & \left\lbrace 
	\begin{array}{l}
  	\max \left\lbrace \begin{array}{l}
  	\theta_3(\kappa, \rho), \\
  	p(J[1, \kappa]) - d_\kappa 
  	\end{array} \right\rbrace, \\
  	\left\lbrace 
	  \begin{array}{ll}
  	  \theta_3(\kappa, \rho - p_\kappa) + p_\kappa
	    & \text{if } \kappa \in J^o\\
	  \infty & \text{if } \kappa \in J^r
	\end{array}	\right. \\
	\end{array} \!\!\!\!\!\! \right\rbrace.
	\end{align*}
	
	In order to obtain the recursion of $\theta_4$, we first prove the following claim on the structure of the function $g'$.	
	\begin{claim}\label{claim:structure-maxl-Y}
	  Let $Y \subseteq H$ and $\kappa < \min Y$. Then 
	  \begin{align*}
	    g'_\kappa(Y) & = \max \{g'_{\kappa + 1}(Y),\; p(J[1, \kappa]) - d_\kappa\} \text{ and }\\
	    g'_\kappa(Y \cup \{\kappa\}) & = \max \{g'_{\kappa + 1}(Y),\; p(J[1, \beta]) - p(Y) - d_\kappa\}.
	  \end{align*}
  \end{claim} 
  \begin{proof}
    The first identity follows from the definition of $g'_\kappa$ and the fact that all jobs $j < \kappa$ precede $\kappa$ in $\sigma_{\emptyset, Y}$ because $\kappa < \min Y$.
    
    Now let $j^* \in J[\kappa, \beta]$ be such that $ g'_\kappa(Y \cup \{\kappa\}) = C^{\sigma_{\emptyset, Y \cup \{\kappa\}}}_{j^*} - d_j^*$. 
    Note that $d_{j^*} \geq d_\kappa$ because jobs are ordered according to EDD.
    Therefore, our choice of $j^*$ implies $C^{\sigma_{\emptyset, Y \cup \{\kappa\}}}_{j^*} \geq C^{\sigma_{\emptyset, Y \cup \{\kappa\}}}_{\kappa}$. We can thus conclude that $j^* \in Y \cup \{\kappa\}$. The second identity then follows from the observation that $C^{\sigma_{\emptyset, Y \cup \{\kappa\}}}_{j} = C^{\sigma_{\emptyset, Y}}_{j}$ for all $j \in Y$ and $C^{\sigma_{\emptyset, Y \cup \{\kappa\}}}_{\kappa} = p(J[1, \kappa] \setminus Y)$, by construction of $\sigma_{\emptyset, Y \cup \{\kappa\}}$ and $\kappa < \min Y$.
  \end{proof}
  
  \Cautoref{claim:structure-maxl-Y} implies the following recursion for $\theta_4$:
	\begin{align*}
  	\theta_4(\beta,\rho) =  &\left\lbrace \begin{array}{ll}
  	p(J[1, \beta]) - d_\beta & \text{ if } \rho = 0 \\
  	\infty & \text{ otherwise}
  	\end{array} \right.,\\
  	\theta_4(\kappa,\rho) = \min &\left\lbrace 
  	\begin{array}{l}
      \max \left\lbrace \begin{array}{l}
    	  \theta_4(\kappa  + 1, \rho), \\
    	  p(J[1, \kappa]) - d_\kappa
	     \end{array} \right\rbrace\\
    \left\lbrace \begin{array}{ll}
    	\max \left\lbrace \begin{array}{l}
    	\theta_4(\kappa + 1, \rho - p_\kappa ), \\
    	J[1, \beta] - (\rho - p_\kappa) - d_\kappa 
    	\end{array} \right\rbrace  & \text{if } \kappa \in J^o\\
    	\infty & \text{if } \kappa \in J^r
    	\end{array}	\right. \\
  	\end{array} \!\!\!\!\!\! \right\rbrace.
	\end{align*}
	
	From the above recursions for $\theta_3$ and $\theta_4$, it is easy to see that we can compute $X_{\kappa, \rho}$ and $Y_{\kappa, \rho}$ for all $\kappa \in J[\alpha, \beta]$ and all $\rho \leq P$ combined in time $\mathrm{O}(nP)$. This concludes the proof of the theorem.
\end{proof}

\subsection{Weighted number of tardy jobs}

\Cautoref{thm:maxlateness} implies that even minimizing the unweighted number of tardy jobs is NP-hard.

\begin{corollary}
  $1|er|\sum U_j$ is NP-hard.
\end{corollary}


In the following, we describe a pseudo-polynomial algorithm for solving $1|er|\sum w_jU_j$. Again, we start with an observation on the structure of an optimal solution.

As in the previous section, we assume the jobs to be numbered according to EDD, i.e., $d_1 \leq \dots \leq d_n$. For disjoint sets $X, Y, Z \subseteq J$, let $\sigma_{X,Y,Z}$ be the sequence consisting of the following five blocks with each block internally ordered according to EDD: The first block consists of the jobs in $X$; the second block consists of the jobs in $Y$; the third block consists of the jobs in $J^r \setminus Y$; the fourth block consists of the jobs in $Z$; the fifth block consists of all remaining jobs.

\begin{lemma}\label{lem:wU-opt-structure}
For each instance of $1|er|\sum w_jU_j$ there exist disjoint sets $X^*, Y^*, Z^* \subseteq J$ such that the sequence $\sigma_{X^*, Y^*, Z^*}$ is optimal and, moreover,
\begin{enumerate}
  \item all jobs in $X^* \cup Y^* \cup Z^*$ are non-tardy 
  in $\sigma_{X^*, Y^*, Z^*}$,
  \item $(X^* \cup Z^*) \cap J^r = \emptyset$, and
  \item $\max (X^* \cup (Y^* \cap J^o))< \min Z^*$.
\end{enumerate} 
\end{lemma}

\begin{proof}
  Let $\sigma$ be an optimal feasible sequence.
  Let $i_1 := \min \{i : \sigma(i) \in J^r\}$ and $i_2 := \max \{i : \sigma(i) \in J^r\}$ be the first and last occurrence, respectively, of an r-job in the sequence.  
  Let $E = \{j \in J \;|\; U^{\sigma}_j = 0\}$ be the set of non-tardy jobs in $\sigma$.
  We define $X = \{j \in E \;|\; \sigma^{-1}(j) < i_1\}$, $Y = \{j \in E \;|\; i_1 \leq \sigma^{-1}(j) \leq i_2\}$, and $Z = \{j \in E \;|\; i_2 < \sigma^{-1}(j)\}$, with $\sigma^{-1}(j)$ being the position of job $j$ in sequence $\sigma$. 
  
  It is easy to see that $X, Y, Z$ are disjoint and fulfill property 2, and that $\sigma_{X, Y, Z}$ is feasible because $Y \cap J^o$ is a subset of  o-jobs processed between the first and the last r-job in $\sigma$, which means 
\begin{align*}
&&p(Y \cup (J^r\setminus Y))\\
&&=p(J^r\cup (Y \cap J^o))\\
&&=p(J^r)+p(Y \cap J^o)\\
&&\leq p(J^r)+p(\{j \in J^o \;|\; i_1 \leq \sigma^{-1}(j) \leq i_2\})\\
&&\leq K^r
\end{align*}   
by feasibility of $\sigma$.
   We now show the following claim, which immediately implies that $X,Y,Z$ fulfills property 1 and that $\sigma_{X,Y,Z}$ is optimal.
   
   \begin{claim}
     Every job in $X \cup Y \cup Z = E$ is non-tardy in $\sigma_{X,Y,Z}$.
   \end{claim}
   \begin{proof}
We obtain $\sigma_{X,Y,Z}$ from $\sigma$ by appliying the following three modifications.

First, we delay all tardy jobs in $J^o$ to the end of the schedule (keeping their relative order). No non-tardy job is delayed by this. Let $i_3$ be the last slot holding a job in $J^r$ after this modification.

Second, we delay all tardy jobs in $J^r$ such that they are sequenced consecutively and the last of them is in position $i_3$ (keeping their relative order). No non-tardy job is delayed by this.

We now have five blocks in the current sequence holding jobs from $X$, $Y$, $J^r\setminus Y$, $Z$, and $J^o\setminus(X,Y,Z)$. Finally, by having each block in EDD we do not cause any currently non-tardy job to be tardy. 
\end{proof}

  
  To establish property 3, let $k = \min Z$ and $S = \{j \in J^o \cap (X \cup Y) \;|\; j > k\}$. Now consider the sequence $\sigma'$ that arises from $\sigma_{X,Y,Z}$ by moving all jobs in $S$ to the position right before job $k$ (in arbitrary order). Note that $C^{\sigma'}_j \leq C^{\sigma_{X,Y,Z}}_j$ for all $j \in J \setminus S$, and $C^{\sigma_{X,Y,Z}}_j \leq C^{\sigma_{X,Y,Z}}_k \leq d_k \leq d_j$ for all $j \in S$. 
  Furthermore, resorting the jobs of $S \cup Z$ in $\sigma'$ according to EDD does not cause any job to become tardy. The resulting sequence is $\sigma_{X^*,Y^*,Z^*}$ for $X^* = X \setminus S$, $Y^* = Y^* \setminus S$, and $Z^* = Z \cup S$ and fulfills all requirements of the lemma.
\end{proof}

\begin{theorem}
	$1|er|\sum w_jU_j$ can be solved in $\mathrm{O}(n P^{4})$ time.
\end{theorem}

\begin{proof}
By \Cautoref{lem:wU-opt-structure}, it is sufficient to identify appropriate sets $X^*, Y^*, Z^*$ as described in the lemma. 
We do so using three dynamic programs: one for constructing candidates for $X^*$ and a prefix of $Y^*$, one for constructing candidates for a suffix of $Y^*$ only containing r-jobs, and one for constructing candidates for $Z^*$.

More precisely, let $\kappa^* = \min Z^* \cup \{n+1\}$, $\rho_1^* = p(X^* \cup (Y^*[1, \kappa-1]))$, and $\rho_2^* = p(X^* \cup Y^* \cup J^r)$.
We enumerate all possible values $\kappa, \rho_1, \rho_2$ for $\kappa^*, \rho_1^*, \rho_2^*$ and determine candidates $X_{\kappa, \rho_1, \rho_2}$ for $X^*$, $Y'_{\kappa, \rho_1, \rho_2}$ for $Y^*[1, \kappa - 1]$, $Y''_{\kappa, \rho_1}$ for $Y^*[\kappa, n]$, and $Z_{\kappa, \rho_2}$ for $Z^*$.
We now describe the three DPs for computing the pair $(X_{\kappa, \rho_1, \rho_2}, Y'_{\kappa, \rho_1, \rho_2})$, and the sets $Y''_{\kappa, \rho_1}$ and $Z_{\kappa, \rho_2}$, respectively. Our goal is to make sure that the jobs in the respective sets will not be tardy while maximizing the total weight of jobs contained the set.

For $\kappa \in \{1, \dots, n+1\}$ and $\rho_1, \rho_2 \in \{0, \dots, P\}$, define
\begin{align*}
  \mathcal{X}_{\kappa,\rho_1,\rho_2} &= \left\lbrace (X, Y') \;\middle|\; 
  \begin{array}{l}
   X, Y' \subseteq J[1, \kappa-1],\ X \cap Y' = \emptyset = X \cap J^r, \\
   \sum_{j' \in X \cup Y' : j' \leq j} p_{j'} \leq d_j, \forall j \in X \cup Y',\\
   p(Y' \cap J^o) \leq K^r - p(J^r), \\
   p(X \cup Y') = \rho_1,\ \rho_1 + p(J^r \setminus Y') = \rho_2
  \end{array}
  \right\rbrace \\
  \mathcal{Y}''_{\kappa,\rho_1} &= \left\lbrace  Y'' \;\middle|\; Y'' \subseteq J^r[\kappa, n], \  \rho_1 + \!\!\! \sum_{j' \in Y'' : j' \leq j}\!\!\! p_{j'} \leq d_j, \forall  j \in Y''\right\rbrace \\
  \mathcal{Z}_{\kappa,\rho_2} &= \left\lbrace  Z \;\middle|\; Z \subseteq J^o[\kappa, n], \  \rho_2 + \!\!\! \sum_{j' \in Z : j' \leq j} \!\!\! p_{j'} \leq d_j, \forall  j \in Z \right\rbrace
\end{align*}
and let 
\begin{align*}
(X_{\kappa,\rho_1,\rho_2},Y'_{\kappa,\rho_1,\rho_2}) \in &\argmax_{(X,Y') \in \mathcal{X}_{\kappa,\rho_1,\rho_2}} w(X), \\
Y''_{\kappa,\rho_1} \in &\argmax_{Y'' \in \mathcal{Y}''_{\kappa,\rho_1}} w(Y''), \text{ and } \\
Z_{\kappa,\rho_2} \in &\argmax_{Z \in \mathcal{Z}_{\kappa,\rho_2}} w(Z)
\end{align*}

Note that $Y''_{\kappa,\rho_1}$ and $Z_{\kappa,\rho_2}$ can be computed in time $\mathrm{O}(nP)$ for all choices of $\kappa, \rho_1, \rho_2$ by a dynamic program for the classic problem $1||\sum w_jU_j$, see \citet{Sahni1976}.

In order to construct $X_{\kappa,\rho_1,\rho_2}$ and $Y'_{\kappa,\rho_1,\rho_2}$, we guess $t = p(X)$ and construct two sets $X, Y'$ by iterating through the jobs from $1$ to $\kappa$ in EDD order, keeping track of the processing time of the jobs added to $X$ so far (denoted by $\varrho$), the processing time of the jobs added to $Y'$ so far ($\varrho'$) and the processing time of the o-jobs added to $Y'$ so far ($\varrho''$). 
To this end, we define
\begin{align*}
  &\theta_{5,t}(k, \varrho, \varrho', \varrho'') = \\
  &\max \left\lbrace  w(X \cup Y')  \;\middle|\;
  \begin{array}{l}
    X, Y' \subseteq J[1, k],\ X \cap Y' = \emptyset = X \cap J^r,   \\
   \sum\limits_{j' \in X : j' \leq j} p_j \leq d_j, \forall  j \in X,\\
   t + \sum\limits_{j' \in Y' : j' \leq j} p_j \leq d_j, \forall j \in Y',\\
   p(X) = \varrho,\ p(Y') = \varrho',\ p(Y' \cap J^o) = \varrho''      
  \end{array}
  \right\rbrace
\end{align*}
and observe that $\theta_{5,t}$ can be computed by the following recursion
\begin{align*}
	&\theta_{5,t}(0, \varrho, \varrho', \varrho'') = \left\lbrace \begin{array}{ll}
	0 & \text{ if } \varrho, \varrho', \varrho'' = 0 \\
	-\infty & \text{ otherwise}
	\end{array} \right., \\
	&\theta_{5,t}(j, \varrho, \varrho', \varrho'') =  \\
	&\max\left\lbrace 
	\begin{array}{l}
   \!\!\! \theta_{5,t}(j-1, \varrho, \varrho', \varrho'')\\
   \!\!\!	\left\lbrace 
	 \begin{array}{ll}
    	\!\!\!\theta_{5,t}(j-1, \varrho - p_j, \varrho', \varrho'') + w_j & \text{ if } \varrho \leq d_j, j \in J^o\\
    	\!\!\!-\infty & \text{ otherwise }
	  \end{array}	\right.\\
  \!\!\!	\left\lbrace 
	  \begin{array}{ll}
  	  \!\!\!\theta_{5,t}(j - \! 1, \varrho, \varrho' -\! p_j, \varrho'' -\! p_j) + w_j & \!\!\text{ if } t + \varrho' \leq d_j, j \in J^o\\
  	  \!\!\!-\infty & \text{ otherwise }
	  \end{array}	\right.\\
  \!\!\!	\left\lbrace 
	  \begin{array}{ll}
  	  \!\!\!\theta_{5,t}(j - 1, \varrho, \varrho' - p_j, \varrho'') + w_j & \text{ if } t + \varrho' \leq d_j, j \in J^r\\
  	  \!\!\!-\infty & \text{ otherwise }
	  \end{array}	\right. 
	\end{array} \!\!\!\!\!\!\!\! \right\rbrace.
\end{align*}
We can thus compute the values $\theta_{5,t}(j, \varrho, \varrho', \varrho'')$ for all choices of $j \in J$ and $t, \varrho, \varrho', \varrho'' \in \{0, \dots, P\}$ in time $\mathrm{O}(nP^4)$. Note that
\begin{align*}
  &w(X_{\kappa,\rho_1,\rho_2}) + w(Y'_{\kappa,\rho_1,\rho_2}) = \\ 
  &\max \left\lbrace \theta_{5,t}(\kappa-1, t, \varrho', \varrho'') \;\middle|\;
  \begin{array}{l} 
  t + \varrho' = \rho_1, \\  \varrho'' \leq K^r - p(J^r), \\
   \rho_1 + p(J^r) - \varrho' + \varrho'' = \rho_2
   \end{array}
  \right\rbrace.
\end{align*}
Hence we can obtain $X_{\kappa,\rho_1,\rho_2}$ and $Y'_{\kappa,\rho_1,\rho_2}$ by iterating through all combinations of $t, \varrho' \in \{0, \dots, P\}$ and $\varrho'' \in \{0, \dots, K^r - p(J^r)\}$. 

After constructing $X_{\kappa,\rho_1,\rho_2}$, $Y'_{\kappa,\rho_1,\rho_2}$, $Y''_{\kappa,\rho_1}$, and $Z_{\kappa,\rho_2}$ for all choices of $\kappa \in \{1, \dots, n+1\}$ and $\rho_1, \rho_2 \in \{0, \dots, P\}$, we can find $\kappa^*, \rho_1^*, \rho_2^*$ so as to maximize $w(X_{\kappa^*,\rho_1^*,\rho_2^*}) + w(Y'_{\kappa^*,\rho_1^*,\rho_2^*}) + w(Y''_{\kappa^*,\rho^*_1}) + w(Z_{\kappa^*,\rho_2^*})$ in time $\mathrm{O}(nP^2)$.
\end{proof}

\section{Complexity results for $1|\gamma|er$ and $1||(\gamma,er)$}
\label{sec:other}

\subsection{Complexity results for $1|\gamma|er$}
\hl{
\begin{theorem}
	\label{totalcompletion2}
	Both $1|\sum C_j|er$ and $1|\max L_j|er$ are NP-hard.
\end{theorem}




We abstain from formal proofs for \Cautoref{totalcompletion2} since they follow easily from NP-hardness of $1|er|\sum C_j$ and $1|er|\max L_j$.
}{R2C3}
\begin{theorem}
	\label{thm:Ar}
	Given $A^\gamma$ as an algorithm that solves $1|er|\gamma$ in $\mathrm{O}(T(n, P,W))$, there is an algorithm ${A}^r$ that solves $1|\gamma|er$ in $\mathrm{O}(T(n,P,W) \log P)$.
\end{theorem}

\begin{proof}  
Note that we can employ the algorithm $A^\gamma$ to check whether there exists a feasible schedule with scheduling cost of at most $K^{\gamma}$ and with a renting period of at most $K^r$. Hence, we can perform binary search to determine the minimum length of the renting period such that there exists a feasible schedule with scheduling cost of at most $K^r$. As $P$ is a natural upper bound on the renting period the binary search takes at most $\log P$ steps.
\end{proof}

In particular, \Cautoref{thm:Ar} implies that we can use any pseudo-polynomial algorithm for $1|er|\gamma$ to obtain a pseudo-polynomial algorithm for $1|\gamma|er$ with the runtime increasing only by a factor of $\log P$.

The generic result of \Cautoref{thm:Ar} suggests that $1|\sum w_j C_j|er$ and $1|\max L_j|er$ are solvable in $\mathrm{O}(nP\min\{W,P\}\log P)$ and $\mathrm{O}(nP\log P)$, respectively. We now show that these two problems can be solved more efficiently. 
\begin{theorem}
$1|\sum w_j C_j|er$ can be solved in $\mathrm{O}(nP\min\{W,P\})$ time.
\end{theorem}

\begin{proof}
Let us obtain subsets ${X}_{\kappa,\rho}$ and ${Y}_{\kappa,\rho}$ for all $\kappa$ with $\alpha < \kappa \leq \beta$ and all $\rho$ with $0 < \rho \le K^r$, as described in the proofs of \Cautoref{lem:twc_nP2} and \Cautoref{lem:twc_nPW}. Then we compute $X^* = X_{\kappa^*,\rho_1^*}$ and $Y^* = Y_{\kappa^*,\rho_2^*}$, where 
\begin{align*}
(\kappa^*,\rho_1^*,\rho_2^*) \in \argmin_{(\kappa,\rho_1,\rho_2) \in \Xi'} \{p(J[\alpha,\beta]) - \rho_1 - \rho_2\}
\end{align*}
and 
\begin{align*}
\Xi' = \{(\kappa,\rho_1,\rho_2) \;|\; {\cal{X}}_{\kappa,\rho_{1}},{\cal{Y}}_{\kappa,\rho_2} \neq \emptyset, \; f_\kappa(X_{\kappa,\rho_1})+ g_\kappa(Y_{\kappa,\rho_2}) \le K^\gamma \}.
\end{align*}

Obtaining all subsets ${X}_{\kappa,\rho}$ and ${Y}_{\kappa,\rho}$ requires $\mathrm{O}(nP\min\{W,P\})$ and finding $(\kappa^*,\rho_1^*,\rho_2^*)$ can be done in $\mathrm{O}(nP)$, using a very similar approach to that in the last paragraph of the proof of \Cautoref{lem:twc_nP2}.
\end{proof}

\begin{theorem}
$1|\max L_j|er$ can be solved in $\mathrm{O}(nP)$ time.
\end{theorem}
\begin{proof}
Let us obtain subsets ${X}_{\kappa,\rho}$ and ${Y}_{\kappa,\rho}$ for all $\kappa$ with $\alpha < \kappa \leq \beta$ and all $\rho$ with $0 < \rho \le K^r$, as described in the proof of \Cautoref{thm:lmax_nP2}. Then we compute $X^* = X_{\kappa^*,\rho_1^*}$ and $Y^* = Y_{\kappa^*,\rho_2^*}$, where 
\begin{align*}
(\kappa^*,\rho_1^*,\rho_2^*) \in \argmin_{(\kappa,\rho_1,\rho_2) \in \Xi'} \{p(J[\alpha,\beta]) - \rho_1 - \rho_2\}
\end{align*}
and 
\begin{align*}
\Xi' = \{(\kappa,\rho_1,\rho_2) \;|\; {\cal{X}}_{\kappa,\rho_{1}},{\cal{Y}}_{\kappa,\rho_2} \neq \emptyset, \; f'_\kappa(X_{\kappa,\rho_1})+ g'_\kappa(Y_{\kappa,\rho_2}) \le K^\gamma \}.
\end{align*}

Obtaining all subsets ${X}_{\kappa,\rho}$ and ${Y}_{\kappa,\rho}$ requires $\mathrm{O}(nP)$ (see the proof of \Cautoref{thm:lmax_nP2}) and finding $(\kappa^*,\rho_1^*,\rho_2^*)$ can be done in $\mathrm{O}(nP)$.
\end{proof}

\subsection{Complexity results for $1||(\gamma,er)$}

\hl{
\begin{theorem}\label{totalcompletion3maxlateness3}
	Both, $1||(\sum C_j,er)$ and $1||(\max L_j,er)$, are NP-hard.
\end{theorem}


Again, we take a pass on a formal proof for \Cautoref{totalcompletion3maxlateness3} since it follows from NP-hardness of $1|er|\sum C_j$ and $1|er|\max L_j$.
}{R2C3}
\begin{theorem}
\label{thm:Ay}
	Given $A^\gamma$ as an algorithm that solves $1|er|\gamma$ in $\mathrm{O}(T(n, P,W))$, there is an algorithm ${A}^{r,\gamma}$ that solves $1||(\gamma,er)$ in  $\mathrm{O}(T(n, P,W) P)$.
\end{theorem}

\begin{proof}
	The proof is similar to the proof of \Cautoref{thm:Ar}.
\end{proof}

The generic result of \Cautoref{thm:Ay} suggests that $1||(\sum w_j C_j,er)$ is solvable in \linebreak $\mathrm{O}(nP^2\min\{W,P\})$. In the following, however, we show that this problem only requires $\mathrm{O}(nP^2)$ to be solved. 

\begin{theorem}
$1||(\sum w_j C_j,er)$ can be solved in $\mathrm{O}(nP^2)$ time.
\end{theorem}
\begin{proof}
Again, let us obtain subsets ${X}_{\kappa,\rho}$ and ${Y}_{\kappa,\rho}$ for all $\kappa$ with $\alpha < \kappa \leq \beta$ and all $\rho$ with $0 < \rho \le K^r$, as described in the proof of \Cautoref{lem:twc_nP2}. Then for each fixed value $L\le K^r$, we compute $X^*_L = X_{\kappa^*_L,\rho_{1,L}^*}$ and $Y^*_L = Y_{\kappa^*_L,\rho_{2,L}^*}$, where 
\begin{align*}
(\kappa^*_L,\rho_{1,L}^*,\rho_{2,L}^*) \in \argmin_{(\kappa,\rho_1,\rho_2) \in \Xi_L} \{f_\kappa(X_{\kappa,\rho_1})+ g_\kappa(Y_{\kappa,\rho_2})\}
\end{align*}
and 
\begin{align*}
\Xi_L = \{(\kappa,\rho_1,\rho_2) \;|\; {\cal{X}}_{\kappa,\rho_{1}},{\cal{Y}}_{\kappa,\rho_2} \neq \emptyset, \;  p(J[\alpha,\beta]) - \rho_1 - \rho_2 = L \}.
\end{align*}

The Pareto front is obtained by considering all sequences $\sigma_{X^*_{L},Y^*_{L}}$. Obtaining all subsets ${X}_{\kappa,\rho}$ and ${Y}_{\kappa,\rho}$ requires $\mathrm{O}(nP^2)$ and finding $(\kappa^*_L,\rho_{1,L}^*,\rho_{2,L}^*)$ for each $L\le K^r$ requires $\mathrm{O}(nP)$ (see the last paragraph of the proof of \Cautoref{lem:twc_nP2}), which combined requires $\mathrm{O}(nP^2)$ time.
\end{proof}

\section{\hl{Complexity results for $1||\gamma+er$}{R2C1}}
\label{sec:comb}
In this section we consider the composite objective $1||\gamma+er$.
Here, we are given a parameter $\lambda \geq 0$ representing the renting cost per time unit and we want to find a sequence minimizing $\gamma(\sigma) + \lambda K(\sigma)$, where $\gamma(\sigma)$ is the cost of $\sigma$ for objective $\gamma$ and $K(\sigma)$ is the length of the rental period. The value $\lambda$ can also be thought of as a parameter chosen by the decision maker to control the trade-off between renting time and the scheduling objective.

We can use the structural results established in \Cautoref{subsec:twc} to solve $1||\sum w_j C_j+er$ in polynomial time. In fact, we can give an explicit description of optimal sequences for any value of $\lambda\geq 0$. This is discussed in \Cautoref{sec:twc-combined}.

We further observe that any optimal solution to $1||\gamma+er$ is included in the Pareto front of the corresponding instance of the bi-objective problem $1||(\gamma,er)$. Therefore, our results from \Cautoref{sec:other} immediately imply pseudo-polynomial algorithms for the objectives $\max L_j$ and $\sum w_j U_j$, as we can find a solution minimizing the composite objective among all Pareto-optimal solutions by simply enumerating the Pareto front.

\begin{observation}
  If $1||(\gamma,er)$ can be solved in $\mathrm{O}(T(n,P,W))$ time, then $1||\gamma+er$ can be solved in $\mathrm{O}(T(n,P,W))$ time.
\end{observation}

We complement this observation in \Cautoref{sec:Lmax-combined} by observing that the reduction used in the proof of \Cautoref{thm:maxlateness} implies NP-hardness of $1||\max L_j+er$ and $1||\sum U_j+er$.

\subsection{Strongly polynomial time algorithm for $1||\sum w_j C_j+er$}
\label{sec:twc-combined}

As in \Cautoref{subsec:twc}, we assume that the jobs are indexed according to WSPT (i.e., $w_1 / p_1 \geq \dots \geq w_n / p_n$) and let $\alpha = \min J^r$, $\beta = \max J^r$. Recall the definition of the sequence $\sigma_{X,Y}$ for disjoint sets $X, Y \subseteq H := J^o[\alpha, \beta]$ introduced in \Cautoref{subsec:twc}: The sequence consists of five blocks and each block is sorted internally according to WSPT. The first block is $J[1, \alpha-1]$; the second block is $X$; the third block is $J[\alpha, \beta] \setminus (X \cup Y)$; the fourth block is $Y$; the fifth block is $J[\beta+1, n]$.

Due to \Cautoref{lem:twc-opt-structure} for every bound on the length of the renting period there exist $X^*, Y^* \subseteq H$ with $\max X^* < \min Y^*$ such that $\sigma_{X^*, Y^*}$ is optimal for the corresponding instance of $1|er|\sum w_jC_j$. Hence, for each point in the Pareto front of the corresponding instance of $1||(\sum w_jC_j,er)$ there are two such sets. Since we find optimum solutions to $1||\sum w_jC_j+er$ among the points of the Pareto front the following lemma is implied.

\begin{lemma}
\label{lem:combined-opt-structure}
	For each instance of $1||er+\sum w_jC_j$ there exists $X^*, Y^* \subseteq H$ with $\max X^* < \min Y^*$ such that $\sigma_{X^*, Y^*}$ is optimal.
\end{lemma}

In the following, we will establish that we can restrict ourselves to sets $X$ and $Y$ with a particular structure when aiming for optimal solutions to $1||\sum w_j C_j+er$. This restriction, then, enables us to solve the problem in polynomial time, contrasting the NP-hardness of $1|er|\sum C_j$ established in \Cautoref{totalcompletion}.

It remains to construct corresponding sets $X, Y \subseteq H$ such that $\sigma_{X,Y}$ is optimal. We do this as follows. For $j \in H$ let $A_j := J^r[\alpha, j]$ and $B_j = J^r[j, \beta]$, i.e., $A_j$ and $B_j$ are the sets of r-jobs with lower or higher index than $j$, respectively. We define 
\begin{align*}
X_{\lambda} & := \left\{ j \in H : \frac{w_j}{p_j} > \frac{w(A_j) - \lambda}{p(A_j)} \text{ and } \frac{w_j}{p_j} \geq \frac{w(J^r)}{p(J^r)}\right\},\\
Y_{\lambda} & := \left\{ j \in H :  \frac{w_j}{p_j} < \frac{w(B_j) + \lambda}{p(B_j)} \text{ and } \frac{w_j}{p_j} < \frac{w(J^r)}{p(J^r)} \right\}.
\end{align*}
We will show in \Cautoref{lem:Xlambda-optimal} that the sequence $\sigma_{X_{\lambda}, Y_{\lambda}}$ is optimal for the composite objective with unit rental cost $\lambda$. Before we can establish this optimality, we derive two additional structural results.

\begin{lemma}\label{lem:combined-monotonicity}
  If $j \in X_{\lambda}$, then $J^o[\alpha, j] \subseteq X_{\lambda}$. If $j \in Y_{\lambda}$, then $J^o[j, \beta] \subseteq Y_{\lambda}$.
\end{lemma}

\begin{proof}
  We only show the first statement of the lemma. The second follows by a symmetric argument.
  
  Let $j' \in J^o[\alpha, j]$.
  Note that $j' \leq j$ and hence $A_{j'} \subseteq A_{j}$ and $w_j / p_j \leq w_{j''}/p_{j''}$ for all $j'' \in A_{j} \setminus A_{j'}$. We obtain $w_{j}p(A_j \setminus A_{j'}) \leq w(A_j \setminus A_{j'}) p_{j}$. Note further that $w_{j} p(A_j) > p_{j} (w(A_j) - \lambda)$ because $j \in X_{\lambda}$. Subtracting the former inequality from the latter, we obtain
  $$w_j p(A_{j'}) = w_{j} (p(A_j) - p(A_j \setminus A_{j'})) > p_{j} (w(A_j) - \lambda - w(A_j \setminus A_{j'})) = p_j (w(A_{j'}) - \lambda).$$
  Note that this implies $w_{j'} / p_{j'} \geq w_j / p_j > (w(A_{j'}) - \lambda) / p(A_{j'})$. Because, furthermore, $w_{j'} / p_{j'} \geq w_j / p_j \geq w(J^r) / p(J^r)$, we conclude $j' \in X_{\lambda}$.
\end{proof}

\Cautoref{lem:combined-monotonicity} reveals the particular structure of sets $X$ and $Y$ we restrict ourselves to: there are jobs $j^X\in J^o$ and $j^Y\in J^o$ such that $X_{\lambda}=J^o[\alpha,j^X]$ and $Y_{\lambda}=J^o[j^Y,\beta]$. Intuitively speaking, the first conditions in the definitions of $X_{\lambda}$ and $Y_{\lambda}$ check whether it is beneficial to have a job $j\in H$ moved before the rental period and after the rental period, respectively, rather than having it in the rental period (as implied by WSPT). Note that both might be beneficial, that is $(w(B_j) + \lambda)/p(B_j)>w_j/p_j> (w(A_j) - \lambda)/p(A_j)$ is possible. If there is such a job $j$ fulfilling both inequalities above, then $X_{\lambda}\cup Y_{\lambda}=H$ and it is not immediately clear whether $j$ should go before or after the rental period. In this case the second condition becomes relevant: since only jobs in $J^r$ are in the rental period, a simple comparison with $w(J^r)/p(J^r)$ serves as a tiebreaker.

It remains to show that $\sigma_{X_{\lambda},Y_{\lambda}}$ is indeed optimal which is accomplished by the next two lemmas.

\begin{lemma}\label{lem:combined-containment}
  If $\sigma_{X^*, Y^*}$ is an optimal sequence for some $X^*, Y^* \subseteq H$, then $X_{\lambda} \subseteq X^*$ and $Y_{\lambda} \subseteq Y^*$.
\end{lemma}

\begin{proof}
  By contradiction assume that $X_{\lambda} \setminus X^* \neq \emptyset$ and let $j = \min X_{\lambda} \setminus X^*$. 
  Note that $J^o[\alpha, j] \subseteq X_{\lambda}$ by \Cautoref{lem:combined-monotonicity} and therefore $J^o[\alpha, j] \setminus \{j\} \subseteq X^*$ by choice of $j$. Hence $j$ is directly preceded in $\sigma^*$ by a block consisting exactly of the jobs in $A_j = J^r[\alpha, j]$ (ordered by their index). Let $\sigma'$ be the sequence obtained from $\sigma^*$ by moving $j$ to the position directly before $\alpha$. Observe that this decreases the renting cost by $\lambda p_j$ and that further
  \begin{align*}
    TWC(\sigma') - TWC(\sigma^*) =  w(A_j)p_j - w_jp(A_j) < \lambda p_j
  \end{align*}
  where the final inequality follows from $j \in X_{\lambda}$. This contradicts the optimality of $\sigma^*$.
  Hence $X_{\lambda} \subseteq X^*$. The statement $Y_{\lambda} \subseteq Y^*$ follows by a symmetric argument.
\end{proof}

\begin{lemma}\label{lem:Xlambda-optimal}
$\sigma_{X_{\lambda},Y_{\lambda}}$ is an optimal sequence for $1||\sum w_j C_j+er$ with unit rental cost $\lambda$.
\end{lemma}

\begin{proof}
  By \Cautoref{lem:combined-opt-structure} there exist $X^*, Y^* \subseteq H$ such that $\sigma^* := \sigma_{X^*, Y^*}$ is optimal. Without loss of generality, we choose $X^*, Y^*$ so as to minimize $|X^* \setminus X_{\lambda}| + |Y^* \setminus Y_{\lambda}|$. 
  
  We first show that $X^* = X_{\lambda}$. Note that $X_{\lambda} \subseteq X^*$ by \Cautoref{lem:combined-containment}. By contradiction assume that $X^* \setminus X_{\lambda} \neq \emptyset$ and let $j = \max X^* \setminus X_{\lambda}$ (see \Cautoref{fig:erplusgammaa}).
  
  \begin{claim}\label{claim:Aj-property}
    $w_j / p_j \leq (w(A_j) - \lambda) / p(A_j)$
  \end{claim}
  \begin{proof}
    By contradiction assume $w_j / p_j > (w(A_j) - \lambda) / p(A_j)$. Note that this implies $w_j / p_j < w(J^r) / p(J^r)$ because $j \notin X_\lambda$.
  Further note that $j \notin Y_{\lambda}$ because $j \in X^* \subseteq H \setminus Y^*$ and $Y_{\lambda} \subseteq Y^*$ by \Cautoref{lem:combined-containment}. Hence $w_j / p_j \geq (w(B_j) + \lambda) / p(B_j)$.
  However, this implies that $w_j / p_j > (w(A_j) + w(B_j)) / (p(A_j) + p(B_j)) = w(J^r) / p(J^r)$, a contradiction.
  \end{proof}
  
  \begin{claim}\label{claim:better-sequence}
    $\sigma_{X^* \setminus \{j\}, Y^*}$ is an optimal sequence, i.e., $TWC(\sigma_{X^* \setminus \{j\}, Y^*}) = TWC(\sigma^*)$.
  \end{claim}
  \begin{proof}
   Let $H' := J^o[\alpha, j - 1] \setminus X^*$ and consider the sequence $\sigma_{X^* \setminus \{j\}, Y^*}$, which arises from $\sigma^*$ by moving $j$ to the position directly after the block $A_j \cup H'$ that it precedes in $\sigma^*$ (see \Cautoref{fig:erplusgammab}).
  Note that $w_{j'} / p_{j'} \geq w_j / p_j$ for all $j' \in H'$ and hence $w(H') p_j \geq p(H') w_j$. Combining this with \Cautoref{claim:Aj-property}, we obtain $w_j (p(A_j) + p(H')) \leq (w(H') + w(A_j) - \lambda) p_j$.
  Hence
  \begin{align*}
    TWC(\sigma_{X^* \setminus \{j\}, Y^*}) - TWC(\sigma^*) =  w_j (p(A_j) + p(H')) - (w(A_j) + w(H'))p_j \leq -\lambda p_j.
  \end{align*}
  Because the renting periods for the two sequences fulfill $K(\sigma_{X^* \setminus \{j\}, Y^*}) - K(\sigma_{X^*, Y^*}) = p_j$, we conclude that $\sigma_{X^* \setminus \{j\}, Y^*}$ is also an optimal sequence.
  \end{proof}
  
  Note that \Cautoref{claim:better-sequence} yields a contradiction to the choice of $X^*, Y^*$ as minimizer of $|X^* \setminus X_{\lambda}| + |Y^* \setminus Y_{\lambda}|$.
  We thus conclude that $X^* = X_{\lambda}$. By a symmetric argument, we can show $Y^* = Y_{\lambda}$ and hence the sequence $\sigma_{X_{\lambda}, Y_{\lambda}}$ is optimal.
\end{proof}

	\begin{figure}
	\centering
		\begin{subfigure}[b]{\linewidth}
		\centering
		\begin{tikzpicture}[
		axis/.style={thick, ->, >=stealth'},
		important line/.style={very thick},
		dashed line/.style={dashed,  thick},
		every node/.style={color=black,}
		]

		\node[fill=green!10, rectangle, draw=black, minimum height=.5cm,minimum width=2cm] at (1.5,0.5) {\scriptsize$X^* \setminus \{j\}$};
		\node[preaction={fill=green!10}, fill=green!10,  pattern=north west lines, pattern color=red!10, rectangle, draw=black, minimum height=.5cm,minimum width=2cm] at (4,0.5) {};
		\node[fill=green!10, rectangle, draw=black, minimum height=.5cm,minimum width=2.5cm] at (6.25,0.5) {\scriptsize $Y^*$};
		\node[fill=green!10, rectangle, draw=black, minimum height=.5cm,minimum width=.5cm] at (2.75,0.5) {\scriptsize$j$};

		\draw[dotted] (.5,1) -- (.5,-0.5) node(yline)[below] {\scriptsize$t_{\alpha}$};
		\draw[dotted] (7.5,1) -- (7.5,-0.5) node(yline)[below] {\scriptsize$t_{\beta}$};

		\draw[decoration={brace,raise=5pt},decorate]
		(.5,.7) -- node[above=6pt] {\scriptsize $X^*$} (3,.7);

		\draw[axis] (-.1,0)  -- (8.5,0);
		\end{tikzpicture}
		\caption{The position of job $j$ in sequences $\sigma_{X^*, Y^*}$.}
		\label{fig:erplusgammaa}
	\end{subfigure}
	\vfill\vfill

	\begin{subfigure}[b]{\linewidth}
		\centering
		\begin{tikzpicture}[
		axis/.style={thick, ->, >=stealth'},
		important line/.style={very thick},
		dashed line/.style={dashed,  thick},
		every node/.style={color=black,}
		]

		\node[fill=green!10, rectangle, draw=black, minimum height=.5cm,minimum width=2cm] at (1.5,0.5) {\scriptsize$X^* \setminus \{j\}$};
		\node[preaction={fill=green!10}, fill=green!10,  pattern=north west lines, pattern color=red!10, rectangle, draw=black, minimum height=.5cm,minimum width=2.5cm] at (3.75,0.5) {};
		\node[fill=green!10, rectangle, draw=black, minimum height=.5cm,minimum width=2.5cm] at (6.25,0.5) {\scriptsize $Y^*$};
		\node[fill=green!10, rectangle, draw=black, minimum height=.5cm,minimum width=.5cm] at (4,0.5) {\scriptsize$j$};

		\draw[dotted] (.5,1) -- (.5,-0.5) node(yline)[below] {\scriptsize$t_{\alpha}$};
		\draw[dotted] (7.5,1) -- (7.5,-0.5) node(yline)[below] {\scriptsize$t_{\beta}$};

		\draw[decoration={brace,raise=5pt},decorate]
		(2.5,.7) -- node[above=6pt] {\scriptsize $H'\cup A_j$} (3.75,.7);

		\draw[axis] (-.1,0)  -- (8.5,0);
		\end{tikzpicture}
		\caption{The position of job $j$ in sequences $\sigma_{X^* \setminus \{j\}, Y^*}$.}
		\label{fig:erplusgammab}
	\end{subfigure}

	\caption{The position of job $j = \max X^* \setminus X_\lambda$ in sequences $\sigma_{X^*, Y^*}$ and $\sigma_{X^* \setminus \{j\}, Y^*}$ in \Cautoref{lem:Xlambda-optimal}.}
	\label{fig:erplusgamma}
\end{figure}
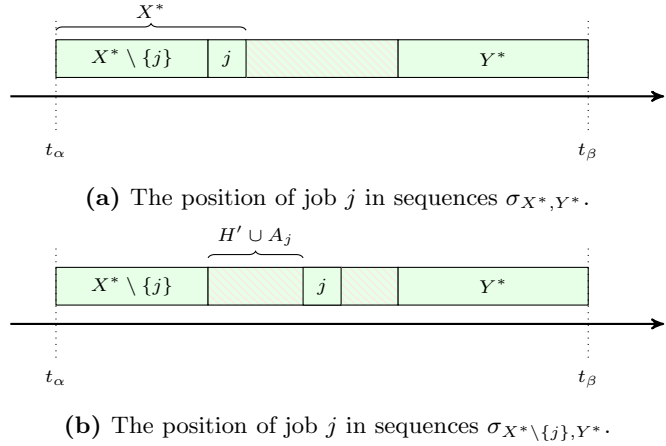

\begin{theorem}
  $1||er+\sum w_jC_j$ can be solved in time $\mathrm{O}(n \log n)$.
\end{theorem}
\begin{proof}
By the above analysis, it suffices to compute the sets $X_{\lambda}, Y_{\lambda}$ and construct the induced sequence. This can be done in linear time once the jobs are ordered according to WSPT. Sorting the jobs can be done in time $\mathrm{O}(n \log n)$.
\end{proof}

\begin{remark}
In fact, our analysis implies that we can obtain optimal solutions for \emph{all} values of $\lambda$ without increasing the computational effort: Because $X_{\lambda} \subseteq X_{\lambda'}$ and $Y_{\lambda} \subseteq Y_{\lambda'}$ for $\lambda < \lambda'$, there are at most $n$ values of $\lambda$ for which the solution computed by the algorithm changes. The corresponding threshold values for $\lambda$ can be obtained in linear time once the jobs are sorted. We can thus efficiently enumerate solutions representing the entire lower convex envelope of the Pareto front.
\end{remark}

\subsection{Complexity results for $1||er+\max L_j$ and $1||er+\sum U_j$}
\label{sec:Lmax-combined}

The hardness for these two variants follows immediately from the construction for the hardness result in \Cautoref{thm:maxlateness}. 

\begin{theorem}
$1||er+\max L_j$ and $1||er+\sum U_j$ are NP-hard.
\end{theorem}

\begin{proof}
  Consider the instance of $1||\max L_j$ constructed in the proof of \Cautoref{thm:maxlateness}. Recall that for this instance $\max_j L_j^{\sigma} \geq 0$ for any sequence $\sigma$ as no due date exceeds the makespan. Furthermore, the corresponding instance of \textsc{Partition} is a yes instance if and only if there exists a sequence $\sigma$ with $\max_j L_j^{\sigma} = 0$ and $K(\sigma) \leq K^r = B+2$.
  Note that solving the corresponding instance of $1||er+\max L_j$ with unit rental cost $\lambda = \frac{1}{B + 3}$ will yield such as sequence if it exists: Indeed, $L_j^{\sigma} = 0$ and $K(\sigma) \leq K^r = B+2$ implies $\max_j L^j + \lambda K(\sigma) \leq \frac{B+2}{B+3} < 1$. As all processing times are integer, any optimal solution for the composite objective must have maximum lateness zero and minimize the length of the rental period among all such sequences. An identical argument applies to the $\sum U_j$-objective.
\end{proof}

\section{Summary and conclusion}
\label{sec:conclusion}

We study \hlim{four} classes of single machines scheduling problems with an external resource: a class of problems where the length of the renting period is budgeted and the scheduling cost needs to be minimized, a class of problems where the scheduling cost is budgeted and the length of the renting period needs to be minimized, a class of two-objective problems where both the length of the renting period and the scheduling cost are to be minimized\hlim{, and, finally, a class of problems where total costs, that is scheduling costs plus rental costs, is to be minimized}. For each class, we consider total (weighted) completion time, maximum lateness, or weighted number of tardy jobs as the scheduling cost function. We show that all discussed problems \hlim{but one} are NP-hard in ordinary sense. \hlim{The remaining problem, namely $1||\sum w_j C_j+er$, can be solved in polynomial time.} \Cautoref{tbl:summery} provides a summary of the complexity of the proposed pseudo-polynomial algorithms in this paper. \hl{It remains open whether problems with scheduling costs reflecting total tardiness can be solved in pseudo-polynomial time.}{R1C3}

A natural generalization considers the case where rental intervals have to be determined for multiple distinct resources and each job can only be scheduled when all its required resources are available. This setting constitutes a generalization of the \emph{linear arrangement problem} (LAP; \citet{Adolphson1973,Liu1995}) to hypergraphs: The jobs correspond to the nodes and each set of jobs requiring a specific resource corresponds to a hyperedge. A schedule corresponds to an ordering of the nodes, where each hyperedge incurs a cost proportional to the difference of the latest completion time and the earliest start time of a job within the hyperedge. The LAP is notorious for being computationally challenging both in theory and practice. Still, devising exponential-time exact methods or efficient approximation algorithms for this setting are interesting directions of future research. \hl{Another interesting generalization is to allow for interrupted multiple rental intervals by considering a fixed rental cost (in addition to the time-sensitive rental cost) each time the equipment is rented.}{R2C1} 

\hli{\paragraph{Acknowledgements} We thank two anonymous referees for their helpful comments, in particular, for suggesting the composite objective discussed in \Cautoref{sec:comb}.}





\bibliographystyle{plainnat}
\bibliography{scheduling}
\end{document}